\documentclass[11pt]{article}

\usepackage[T1]{fontenc}

\usepackage[utf8]{inputenc}

\usepackage[margin=1in]{geometry}
\usepackage{amsthm,amsmath,amssymb}

\usepackage{tikz,xcolor,graphicx}
\usepackage{tikz-cd}

\usepackage{extarrows}

\usepackage{fancyhdr}

\usepackage{mathtools}
\usepackage{thmtools}
\usepackage{thm-restate}

\usepackage{booktabs}

\usepackage[noadjust]{cite}

\usepackage{listings}
\usepackage[shortlabels]{enumitem}

\usepackage{comment}

\usepackage{nicefrac}

\usepackage[unicode]{hyperref}
\hypersetup{colorlinks=true,urlcolor=blue,citecolor=blue,linkcolor=blue}
\usepackage[noabbrev,capitalize]{cleveref}
\crefname{equation}{}{}

\usepackage{url}

\usepackage[color, final]{showkeys} 

\colorlet{refkey}{orange!20}
\colorlet{labelkey}{blue!80}

\usepackage{marginnote}

\usetikzlibrary{calc}

\usetikzlibrary{arrows, shapes} 
\tikzstyle{vertex}=[circle,fill=black!25,minimum size=20pt,inner sep=0pt]
\tikzstyle{selected vertex} = [vertex, fill=red!24]
\tikzstyle{edge} = [draw,thick,-]
\tikzstyle{weight} = [font=\small]
\tikzstyle{selected edge} = [draw,line width=5pt,-,red!50]
\tikzstyle{ignored edge} = [draw,line width=5pt,-,black!20]

\usepackage{float}
\usepackage{pgfplots}
\usepgfplotslibrary{fillbetween}
\pgfplotsset{compat=1.15}

\theoremstyle{plain}
\newtheorem{theorem}{Theorem}[section]

\newtheorem{lemma}[theorem]{Lemma}
\newtheorem{claim}[theorem]{Claim}
\newtheorem{corollary}[theorem]{Corollary}

\theoremstyle{definition}
\newtheorem{definition}[theorem]{Definition}
\newtheorem{example}[theorem]{Example}

\theoremstyle{remark}
\newtheorem*{remark}{Remark}

\usepackage{titletoc}
\newif\ifnotes
\notestrue

\newcommand{\eps}{\epsilon}

\newcommand{\cA}{\mathcal{A}}

\newcommand{\cC}{\mathcal{C}}

\newcommand{\cP}{\mathcal{P}}

\newcommand{\paren}[1]{\left( #1 \right)}

\newcommand{\ceil}[1]{\left\lceil #1 \right\rceil}

\renewcommand{\t}{\text}
\newcommand{\f}{\frac}
\newcommand{\interior}[1]{%
  {\kern0pt#1}^{\mathrm{o}}%
}

\renewcommand{\P}{\mathsf{P}}
\renewcommand{\L}{\mathsf{L}}

\newcommand{\NP}{\mathsf{NP}}

\newcommand{\BPP}{\mathsf{BPP}}

\newcommand{\MA}{\mathsf{MA}}

\newcommand{\TS}{\mathsf{TS}}
\newcommand{\QMA}{\mathsf{QMA}}
\newcommand{\QCMA}{\mathsf{QCMA}}

\def \NTIME{{\mathsf{NTIME}}}
\def \BPTIME{{\mathsf{BPTIME}}}

\def \TIME{{\mathsf{TIME}}}

\usepackage{braket}
\usepackage[pdf]{pstricks}
\usepackage{caption}

\newcommand{\EBQP}{\exists \cdot \BQP}
\newcommand{\EBQTIME}{\exists \cdot \mathsf{BQTIME}}
\newcommand{\BQP}{\mathsf{BQP}}
\newcommand{\BQTIME}{\mathsf{BQTIME}}

\newcommand{\BPTS}{\mathsf{BPTS}}
\newcommand{\MATIME}{\mathsf{MATIME}}
\newcommand{\QCMATIME}{\mathsf{QCMATIME}}
\newcommand{\SAT}{\mathsf{SAT}}

\graphicspath{{images/}}

\title{Time-Space Lower Bounds for Simulating Proof Systems with Quantum and Randomized Verifiers}
\author{Abhijit S. Mudigonda \and R. Ryan Williams\footnote{Supported by NSF CCF-1909429 and CCF-1741615.}}
\date{}

\begin{document}    
\maketitle
\begin{abstract}
        A line of work initiated by Fortnow in 1997 has proven model-independent time-space lower bounds for the $\SAT$ problem and related problems within the polynomial-time hierarchy. For example, for the $\SAT$ problem, the state-of-the-art is that the problem cannot be solved by random-access machines in $n^c$ time and $n^{o(1)}$ space simultaneously for $c < 2\cos(\f{\pi}{7}) \approx 1.801$.  

        We extend this lower bound approach to the quantum and randomized domains. Combining Grover's algorithm with components from $\SAT$ time-space lower bounds, we show that there are problems verifiable in $O(n)$ time with quantum Merlin-Arthur protocols that cannot be solved in $n^c$ time and $n^{o(1)}$ space simultaneously for $c < \f{3+\sqrt{3}}{2} \approx 2.366$, a super-quadratic time lower bound. This result and the prior work on $\SAT$ can both be viewed as consequences of a more general formula for time lower bounds against small space algorithms, whose asymptotics we study in full. 
        
        We also show lower bounds against randomized algorithms: there are problems verifiable in $O(n)$ time with (classical) Merlin-Arthur protocols that cannot be solved in $n^c$ randomized time and $n^{o(1)}$ space simultaneously for $c < 1.465$, improving a result of Diehl. For quantum Merlin-Arthur protocols, the lower bound in this setting can be improved to $c < 1.5$. 
\end{abstract}

\thispagestyle{empty}
\addtocounter{page}{-1}
\newpage    

\section{Introduction}
\label{sec:intro}

A flagship problem in computational complexity is to prove lower bounds for the $\SAT$ problem. While it is conjectured that no polynomial-time algorithms exist for $\SAT$ (in other words, $\P \ne \NP$), not much progress has been made in that direction. Furthermore, several significant barriers towards such a separation are known~\cite{relativization,natural,algebrization}. Therefore, approaches have centered around proving weaker lower bounds on $\SAT$ first.

A natural preliminary step in showing that no polynomial-time algorithm can decide $\SAT$ is showing that no algorithms of logarithmic space can decide $\SAT$, or in other words, showing that $\L \ne \NP$. Unlike the $\P$ vs. $\NP$ problem, the aforementioned complexity barriers (arguably) do not apply as readily to $\L$ vs. $\NP$, and concrete progress has been made.\footnote{For example, there exist oracles relative to which the lower bounds in the following paragraph are false.}

Following a line of work~\cite{lvfvm}, R.~Williams \cite{williams-mod} proved that $\SAT$ (equivalently\footnote{At least, up to polylogarithmic factors.} $\NTIME[n]$, nondeterministic linear time) cannot be decided by algorithms (even with constant-time random access to their input and storage) using both $n^{o(1)}$ space and $n^c$ time, for $c < 2\cos(\frac{\pi}{7}) \approx 1.802$. If one could show the same lower bound for arbitrarily large constant $c$, the separation $\L \ne \NP$ would follow immediately. In the following we use $\TS[n^c]$ to denote the class of languages decidable by $n^{o(1)}$-space algorithms using $n^c$ time.

All the aforementioned work builds on the \emph{alternation-trading proofs} approach~\cite{williams-aut}. This approach combines two elements: a \emph{speedup rule} that reduces the runtime of an algorithm by ``adding a quantifier'' ($\exists$ or $\forall$) to an alternating algorithm, and a \emph{slowdown rule} that uses a complexity theoretic assumption (for example, $\SAT \in \TS[n^c]$) to ``remove a quantifier'' and slightly increase the runtime of the resulting algorithm. Both rules yield inclusions of complexity classes. Our ultimate goal is to contradict a time hierarchy theorem (e.g., proving $n^{100}$ time computations can be simulated in $n^{99}$ time) by applying these rules in a nice order, and with appropriately chosen parameters. 

One may hope that the constant $c$ from~\cite{williams-mod} can be made arbitrarily large, and eventually show that $\L \ne \NP$. Unfortunately, in~\cite{williams-buss}, R.~Williams and S.~Buss showed that no alternation-trading proof based purely on the speedup and slowdown rules from that line of work could improve on the exponent of~\cite{williams-mod}.

Nevertheless, there is hope that alternation-trading proofs might yield stronger lower bounds for problems harder than $\SAT$. For example, R. Williams~\cite{williams-aut} showed that the $\Sigma_2 \P$-complete problem $\Sigma_2\SAT$ is not in $\TS[n^c]$, for $c < 2.903$. In this paper, we make further progress in this direction. In particular, we focus on the quantum and randomized analogues of $\NTIME[n]$, $\QCMATIME[n]$ and $\MATIME[n]$, obtaining stronger lower bounds against both classes.\footnote{Recall that $\QCMA$ (quantum classical Merlin-Arthur) is essentially $\NP$ with a quantum verifier and $\MA$ (Merlin-Arthur) is essentially $\NP$ with a randomized verifier.} We believe our lower bound for $\QCMATIME[n]$ (\cref{thm:qcma-ts-lb}) to be particularly interesting because it yields a \emph{nontrivial} separation between a quantum complexity class and a classical complexity class \textit{without the need for oracles}.\footnote{By ``nontrivial'', we mean a separation that does not immediately follow from known classical results. For example, $\QCMATIME[n] \not \subseteq \TS[n^{1.8}]$ follows immediately from the classical lower bound $\NTIME[n] \not \subseteq \TS[n^{1.8}]$, but our result does not.} While there are several results~\cite{bravyi2018quantum,bqp-ph,WattsKST19} demonstrating the power of quantum computation against very restricted low-depth classical circuit models ($\mathsf{NC}^0$, $\mathsf{AC}^0$, $\mathsf{AC}^0[2]$) which also imply strong oracle separation results, our result appears to be the first non-trivial lower bound for a quantum class against the much more general random-access machine model (with simultaneous time and space constraints).

\subsection{Our Results}

\subsubsection{Generic slowdown rules and a lower bound for $\QCMATIME[n]$}

For showing stronger lower bounds on $\QCMATIME[n]$, our key observation is that the stronger assumption $\QCMATIME[n] \subseteq \TS[n^c]$ (compared to $\NTIME[n] \subseteq \TS[n^c]$) can be applied to construct a stronger conditional slowdown rule. Formally, we generalize the previous framework of alternation-trading proofs by introducing the notion of a \emph{generic slowdown rule} (defined formally in \cref{def:gen-slowdown}), which are slowdown rules parameterized by a constant $\alpha \in (0,1]$ controlling the runtime cost associated with removing a quantifier. Smaller values of $\alpha$ correspond to stronger slowdown rules. We prove the following theorem showing how generic slowdown rules imply time-space lower bounds.

\begin{restatable}{maintheorem}{genlb}
    \label{thm:gen-lb}
	Fix $\alpha \in (0,1]$ and let $\cC$ be a complexity class. Let $r_1$ be the largest root of the polynomial $P_{\alpha}(x) := \alpha^2x^3 - \alpha x^2 -2\alpha x +1$. If $\cC \subseteq \TS[n^c]$ implies a generic slowdown rule with parameters $\alpha$ and $c$, then $\cC \not \subseteq \TS[n^c]$ for $c < r_1$. 
\end{restatable}

The assumption $\NTIME[n] \subseteq \TS[n^c]$ implies a slowdown rule with $\alpha = 1$. The previous $\NTIME[n] \not \subseteq \TS[n^c]$ for $c < 2\cos(\f{\pi}{7})$ lower bound~\cite{williams-mod} becomes an immediate corollary of \cref{thm:gen-lb} if we take $\cC = \NTIME[n]$. The stronger slowdown rule that we obtain from the stronger assumption $\QCMATIME[n] \subseteq \TS[n^c]$ has $\alpha = \f{2}{3}$, which allows us to derive lower bounds for larger values of $c$. In particular, we obtain the following lower bound for Quantum (Classical) Merlin-Arthur linear time.

\begin{restatable}{maintheorem}{qcmatslb}
    \label{thm:qcma-ts-lb}
    $\QCMATIME[n] \not \subseteq \TS[n^c]$ for $c < \f{3+\sqrt{3}}{2} \approx 2.366$.
\end{restatable}

The main advantage of the generic slowdown approach and \cref{thm:gen-lb} is that improvements in slowdown rules translate immediately into stronger bounds against $\TS$. \cref{fig:cvsalpha} shows how the lower bound exponent we obtain changes as $\alpha$ does. As expected, the lower bound exponent goes to infinity as $\alpha$ approaches zero, but we see that even modest improvements in $\alpha$ yield substantially stronger bounds. We discuss potential applications further in \cref{subsec:future}. In the appendix, we show that this dependence between the generic slowdown parameter and lower bound exponent is ``optimal'' for the tools we use, extending the optimality theorem of Buss and Williams~\cite{williams-buss} to the general case while also providing a shorter proof of their optimality theorem.
 
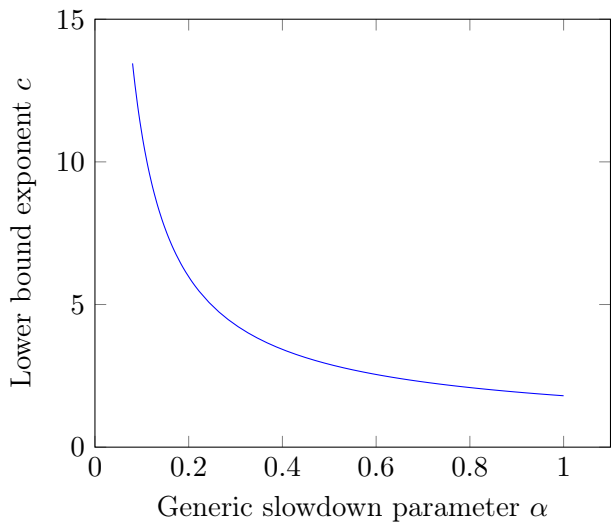
\begin{figure}[H]
    \centering
    \begin{tikzpicture}
        \begin{axis}[
            xlabel={Generic slowdown parameter $\alpha$},
            ylabel={Lower bound exponent $c$},
            xmin=0, xmax=1.1,
            ymin=0, ymax=15
        ]
        \addplot[color=blue,no marks] table [x=x, y=y, col sep = comma] {images/cvsalpha-plot.csv};
        \end{axis}
    \end{tikzpicture}
    \caption{The lower bound exponent as a function of the generic slowdown parameter $\alpha$.}
    \label{fig:cvsalpha}
\end{figure}

\subsubsection{Lower bounds against randomized small space machines}
\label{subsubsec:bpts-lb}

We also prove lower bounds against randomized small space machines. While the techniques used in this setting are similar, the results do not follow from \cref{thm:gen-lb}, so we state them separately. We refer to the class of languages decidable by a $n^{o(1)}$-space randomized machine in $n^c$ time as $\BPTS[n^c]$. We prove lower bounds for both linear-time Merlin-Arthur protocols and linear-time Classical-Merlin Quantum-Arthur protocols.

\begin{restatable}{maintheorem}{bptslb}
        \label{thm:bpts-lb} 
        Let $r_1 \approx 1.465$ be the largest root of the polynomial $x^3-x^2-1$. Then, $\MATIME[n] \not \subseteq \BPTS[n^c]$ for $c < r_1$. 
        Furthermore, $\QCMATIME[n] \not \subseteq \BPTS[n^c]$ for $c < 1.5$.
\end{restatable}

Prior to this work, the state-of-the-art, due to~\cite{dvm, diehl-maam}, was that $\MATIME[n] \not \subseteq \BPTS[n^c]$ for $c < \sqrt{2} \approx 1.414$.  Observe that $\MATIME[t] \subseteq \mathsf{PPTIME}[t^2]$ because we can amplify Arthur's completeness and soundness to $(1-2^{-100t}, 2^{-100t})$ while increasing the runtime of the verifier by a factor of $O(t)$, and we can union bound over Merlin's strings.\footnote{By a union bound, there is a gap between the case where all $2^t$ Merlin strings have a $2^{-100t}$ chance of accepting, and the case where a single Merlin string accepts with probability at least $1-2^{-100t}$.} A similar argument, coupled with the quasilinear-time simulation of bounded-error quantum computation with unbounded-error random computation of \cite{vmw}, shows that $\QCMATIME[t] \subseteq \mathsf{PPTIME}[t^2]$. Therefore, pushing either of the lower bound exponents in \cref{thm:bpts-lb} to beyond $2$ would yield superlinear bounds for decision versions of counting-type problems (e.g. $\mathsf{MAJ}\mbox{-}\SAT$) against randomized small space machines. (Note it is known that $\#\SAT$ requires $\tilde{O}(n^2)$ time on randomized $n^{o(1)}$-space machines~\cite{williams-mckay}.) While these results are admittedly incremental improvements, they use some different ideas compared to previous works, and may be amenable to further improvement (see \cref{subsec:future} for more details).

The lower bounds of \cref{thm:qcma-ts-lb} and \cref{thm:bpts-lb} actually hold for complexity classes that are presumably even ``smaller'' than $\QCMATIME[n]$ and $\MATIME[n]$; we describe these further in \cref{sec:prelim}. 

\subsection{Techniques}

\subsubsection{Alternation-trading proofs}

Many time-space lower bounds for SAT and related problems are proved via \emph{alternation-trading proofs}, which give a chain of inclusions of complexity classes. We will formally define alternation-trading proofs in \cref{sec:prelim}; for now, let us give a cursory explanation. An \emph{alternation-trading proof} consists of a sequence of containments of \emph{alternating complexity classes}. An alternating complexity class can be thought of as a ``fine-grained'' version of $\Sigma_k \P$ or $\Pi_k \P$: it is a complexity class parameterized by $(k+1)$ positive constants bounding the length of the output of each quantifier and the verifier runtime. For example, 
\begin{equation*}
    (\exists n^2) (\forall n^{2}) \TS[n^5] 
\end{equation*}
is an alternating complexity class. This notation refers to the class of languages decided by a $\Sigma_2$ machine where, on inputs of length $n$, the two quantifiers each quantify over $\tilde{O}(n^2)$ bit strings and the \emph{verifier runtime} is $\tilde{O}(n^5)$.

In an alternation-trading proof, there are two main ways of passing from one alternating complexity class to the next. The first is a \textbf{speedup rule}, which adds a quantifier to the class, and decreases the verifier runtime. For example, a speedup rule might yield an inclusion of the form
\begin{equation}
    \label{eq:int-speedup}
    \dots \TS[n^d] \subseteq \dots (Q n^x) (\neg Q x\log n) \TS[n^{d-x}]
\end{equation}
for some constant $0 < x < d$ and quantifier $Q \in \{\exists, \forall\}$, where $\neg Q$ denotes the opposite quantifier and the $\dots$ refer to other quantifiers. Two important points to note are that (a) the speedup rule is generally an unconditionally true inclusion and (b) the second quantifier has only $O(\log n)$ bits.

The second major component of alternation-trading proofs is a \textbf{slowdown rule}, which removes a quantifier and increases the verifier runtime. We will use slowdown rules that hold conditioned on complexity-theoretic assumptions (that we will later contradict). For example, the slowdown rule used to prove lower bounds on $\NTIME[n]$ can be informally stated as follows: assuming $\NTIME[n] \subseteq \TS[n^c]$ for some $c > 0$, 
\begin{equation}
    \label{eq:inf-slowdown}
    \dots (Q n^a) (\neg Q n^b) \TS[n^d] \subseteq \dots (Q n^a) \TS[n^{c\cdot\max(a,b,d)}].
\end{equation}
where again $Q \in \{\exists, \forall\}$ and $\neg Q$ denotes the opposite quantifier. This rule follows from an application of padding/translation.

While the speedup and slowdown rules are themselves simple, the existing lower bounds on $\NTIME[n]$ arise by applying these rules in a long intricate sequence, with appropriately chosen parameters for the speedup rule applications. Ultimately, we aim to use slowdown and speedup rules to exhibit a sequence of inclusions that shows (for example) that $\NTIME[n^d] \subseteq \NTIME[n^{d'}]$ for $d' < d$, contradicting the nondeterministic time hierarchy theorem and demonstrating that our initial assumption must have been false. 

All time-space lower bounds for SAT against random-access models of computation, including the state-of-the-art bound~\cite{williams-mod}, use an alternation-trading proof. This particular proof will be a starting point for this work. 

\subsubsection{Generic slowdown rules}

We start by introducing the notion of a generic slowdown rule. Generic slowdown rules are parameterized by a constant $\alpha$ such that $0 < \alpha \leq 1$, along with a constant $c \geq 1$ that generally comes with an assumption being made. Informally, generic slowdown rules allow us to --- under an appropriate assumption --- prove conditional inclusions of the form
\begin{equation*}
    \dots (Q n^a) (\neg Q n^b) \TS[n^d] \subseteq \dots (Q n^a) \TS[n^{\alpha c\cdot\max(a,b,d)}].
\end{equation*}

Observe that when $\alpha = 1$ we recover \cref{eq:inf-slowdown}, but when $\alpha < 1$ we obtain stronger inclusions. In \cref{thm:gen-lb}, we use generic slowdowns in the alternation-trading proof from \cite{williams-mod} and characterize the lower bound we obtain as a function of the parameter $\alpha$ in our generic slowdown rule. While the core idea of this proof is essentially the same as the presentation of the proof in \cite{williams-aut} (and the result can be thought of as ``putting $\alpha$ everywhere''), our proof technique is somewhat different. In the appendix, this different approach yields a shorter proof of optimality than the one presented by Buss and Williams~\cite{williams-buss}. 

\subsubsection{A generic slowdown rule from Grover search}

In order to apply \cref{thm:gen-lb} to $\cC = \QCMATIME[n]$ and obtain a better lower bound for $\QCMA$, we show that the assumption $\QCMATIME[n] \subseteq \TS[n^c]$ implies a generic slowdown rule for $\alpha = \f{2}{3}$. Recall that Grover's algorithm lets us search a space of size $N$ with only $O(\sqrt{N})$ quantum queries. We obtain our stronger slowdown rule by showing that Grover's algorithm can be used to more efficiently remove the $(x \log n)$-bit quantifiers that arise after applications of the speedup rule, such as \eqref{eq:int-speedup}. In the $\NTIME$ vs $\TS$ setting, there are two ways to remove an $(x \log n)$-bit quantifier. First, we could remove it with an $O(n^x)$ multiplicative blowup, by having the verifier exhaustively search through all strings of length $x\log n$. However, naively running $n^x$ trials of an $n^{d-x}$ computation would take $n^d$ time, and our simulation would end up no faster than it was initially. Second, we may try to use a slowdown rule as in \eqref{eq:inf-slowdown}, but this incurs a runtime cost that depends on $c$. This option becomes expensive as we try to increase $c$ and prove stronger lower bounds. Our key insight is that if our verifier is allowed to be quantum, Grover's algorithm can be applied to perform this quantifier elimination in only $O(n^{x/2})$ extra time overhead, independent of $c$. Then, by applying the assumption $\QCMATIME[n] \subseteq \TS[n^c]$, we can remove this quantum verifier along with the quantifier $(\exists n^x)$ and ultimately demonstrate the inclusions

\begin{align*}
    & \dots (Q n^a) (\neg Q n^b) \TS[n^d] \\
    \subseteq & \dots (Q n^a) (\neg Q n^b) (Q n^x) (\neg Q x\log n) \TS[n^{d-x}] & \t{Speedup Rule} \\
    \subseteq & \dots (Q n^a) (\neg Q n^b) (Q n^x) \BQTIME[n^{d-\f{x}{2}}] & \t{Grover's Algorithm}\\
    \subseteq & \dots (Q n^a) (\neg Q n^b) \TS[n^{c \cdot \max(b, x, d-\f{x}{2})}]. & \t{Assumption on $\QCMATIME[n]$}\\ 
\end{align*}

Letting $x := \f{2d}{3}$, we obtain a generic slowdown rule with $\alpha = \f{2}{3}$.

\subsubsection{Lower bounds against randomized small space machines}

Some obstacles arise when trying to prove \cref{thm:bpts-lb}, which shows lower bounds against $\BPTS$.
The main problem is that the usual speedup rules for deterministic computation do not tell us how to use quantifiers to speedup randomized small space computations. Fortunately, this particular issue was resolved by Diehl and Van Melkebeek \cite{dvm}, who gave a speedup rule for small space randomized machines by coupling Nisan's space-bounded derandomization \cite{nisan} with the Sipser-G\'{a}cs-Lautemann theorem \cite{sgl}. This speedup rule, while somewhat less efficient than the speedup rule for deterministic machines, is still enough to obtain interesting superlinear time lower bounds. Applying this speedup rule, similar arguments as used in \cref{thm:gen-lb} yield the desired lower bounds.

\subsection{Future Work}
\label{subsec:future}

As mentioned earlier, the main advantage of the generic slowdown framework and \cref{thm:gen-lb} is that improvements in slowdown rules translate immediately into stronger bounds against $\TS$. To this end, we highlight two particularly interesting directions.

\begin{enumerate}
\item Is it possible to prove a ``quantum speedup rule'', whereby the runtime of quantum computations could be reduced by adding quantifiers (possibly over quantum states)? Presently, we are forced to use a slowdown rule to remove a quantum verifier from an alternating complexity class as soon as it is added. Having a quantum speedup rule would enable us to work with the quantum verifier before removing it, drastically widening the scope for potential alternation-trading proofs. It's not hard to show that even certain weak forms of a quantum speedup rule would improve the generic slowdown parameter $\alpha$ we can obtain in the $\QCMA$ vs.\ $\TS$ setting. Speedup rules also have applications outside complexity theory. For example, versions of speedup rules for low-space computation appear in the construction of delegation schemes in cryptography \cite{bitansky-rec-comp, kpy-delegation19, kpy-delegation20}, and quantum speedup rules could play a part in extending such work to the quantum domain. 

\item Can we use the ideas of this paper to improve existing lower bounds on counting-type ($\#\P$ related) problems, such as $\#\SAT$ and $\mathsf{MAJ}\mbox{-}\SAT$? For example, could our super-quadratic time lower bound for $\QCMATIME[n]$ be somehow applied to obtain super-quadratic lower bounds for $\#\SAT$ as well? Because \[\MATIME[t] \subseteq \QCMATIME[t] \subseteq \mathsf{PPTIME}[t^{2+o(1)}]\] as discussed in \cref{subsubsec:bpts-lb}, lower bounds on $\QCMATIME[n]$ and $\MATIME[n]$ do translate to some lower bounds for counting problems against small space. However, the known reductions from classes like $\MA$ and $\QCMA$ to counting problems incur a quadratic blowup. Furthermore, there is evidence that a quadratic blowup is necessary for black-box techniques~\cite{diehl-maam,watson-mapp}. As such, it appears we must either improve the lower bound exponent, or prove that we can bypass the quadratic blowup outside of black-box settings. 
\end{enumerate}

\subsection{Organization}
\cref{sec:prelim} covers relevant background, especially on alternation-trading proofs. In~\cref{sec:gen-lb}, we study alternation-trading proofs with generic slowdowns and prove~\cref{thm:gen-lb}. In~\cref{sec:grodown}, we use Grover's algorithm to prove~\cref{lem:grodown}, allowing us to obtain a generic slowdown with $\alpha = \f{2}{3}$ and prove~\cref{thm:qcma-ts-lb}. In~\cref{sec:bpts}, we prove~\cref{thm:bpts-lb}. In Appendix~\ref{sec:limits}, we prove that~\cref{thm:gen-lb} is optimal among alternation-trading proofs using the prescribed speedup and slowdown rule.

\subsection{Acknowledgments}

The first author thanks R. Williams for his support and patience throughout this research. The first author also thanks Aram Harrow, Peter Shor, Saeed Mehraban, Ashwin Sah, and Lisa Yang for contributing office space and helpful conversations and Lijie Chen and Shyan Akmal for reading and editing a draft of this manuscript. Lastly, the first author thanks the Theory Group at MIT CSAIL for providing him with free food and for not revoking his card access. He is sorry he ate so many of the chocolate-covered pretzels.

\section{Preliminaries}
\label{sec:prelim}

We assume familiarity with classical complexity and quantum computing on the levels of \cite{arora-barak} and \cite{nielsen-chuang}, respectively. 

\subsection{Alternation-Trading Proofs}
\subsubsection{A Simple Example: The Lipton-Viglas Bound}

We introduce alternation-trading proofs by describing the Lipton-Viglas bound. The proof captures many of the features of alternation-trading proofs in general. 

\begin{theorem}[\cite{lip-vig}]
    \label{thm:lip-vig}
    For all $c < \sqrt{2}$,
    \begin{equation*}
        \NTIME[n] \not \subseteq \TS[n^c].
    \end{equation*}
\end{theorem}

We will omit the proofs of the Lipton-Viglas speedup and slowdown rules, as we discuss the general forms, \cref{cor:speedup-rule} and \cref{cor:slowdownrule}, later on in the paper.
\begin{lemma}[Lipton-Viglas Speedup Rule]
    For any $d > 0$,
    \begin{equation*}
        \TS[n^d] \subseteq \Sigma_2\mathsf{TIME}[n^{\f{d}{2}}]
    \end{equation*}
\end{lemma}
At a high level, this speedup rule follows because we can chop an $n^d$-step computation into $\sqrt{n^d}$ pieces of length $\sqrt{n^d}$ and verify each piece independently if we know the intermediate machine states.

\begin{restatable}[Slowdown Lemma]{lemma}{usualslowdown}
    \label{lem:usual-slowdown}
    Assume that $\NTIME[n] \subseteq \TS[n^c]$ for some $c > 1$. Then, for any $d \geq 1$, 
    \begin{equation*}
        \mathsf{NTIME}[n^d] \cup \mathsf{coNTIME}[n^d] \subseteq \TS[n^{cd}].
    \end{equation*}
\end{restatable}

The slowdown rule follows from a padding argument and the observation that $\TS[n^c]$ is closed under complementation. 

\begin{proof}[Proof of \cref{thm:lip-vig}]
    Applying the speedup rule once and the slowdown rule twice, we find that
    \begin{align*}
        \Sigma_2\mathsf{TIME}[n^2] & \subseteq (\exists n^2)\mathsf{coNTIME}[n^2] & \t{Definition of $\Sigma_2$} \\
        & \subseteq \exists \TS[n^{2c}] & \t{Slowdown Rule} \\
        & \subseteq \TS[n^{2c^2}] & \t{Slowdown Rule} \\
        & \subseteq \Sigma_2\mathsf{TIME}[n^{c^2}]. & \t{Speedup Rule} 
    \end{align*}
By the time hierarchy theorem (for $\Sigma_2$ machines), we have a contradiction if $c^2 < 2$. 
\end{proof}

\subsection{Alternation-trading proofs in general}

We will start by defining various time-space complexity classes to this work. 

\begin{definition}
    $\TS[t(n)]$ is the class of languages computable by a deterministic random-access machine using space $n^{o(1)}$ and time $O(t(n)^{1+o(1)})$ on an $n$-bit input. $\BPTS[t(n)]$ is the class of languages computable by a two-sided error randomized random-access machine using space $n^{o(1)}$ and time $O(t(n)^{1+o(1)})$ on an $n$-bit input.
\end{definition}

Note that a \emph{randomized} random-access machine with random access to its input has only read-once access to its randomness. 

\begin{definition}
    \label{def:alt-class}
    For positive constants $\{a_i\}_{i \geq 1}$ and $\{b_i\}_{i \geq 1}$ and quantifiers $Q_i \in \{\exists, \forall\}$, the \textbf{alternating complexity class} $(Q_1 n^{a_1})^{b_1} (Q_2 n^{a_2})^{b_2} \dots (Q_k n^{a_k})^{b_k} \TS[n^d]$ is the set of languages decidable by a machine operating in the following fashion on an $n$-bit input. Computation occurs in $k+1$ stages. In the $i^{\t{th}}$ stage, for $1 \leq i \leq k$, the machine obtains from $Q_i$ a string of length $n^{a_i+o(1)})$. It then uses an $n^{o(1)}$-space machine and $n^{b_i+o(1)}$ time to compute $n^{b_i+o(1)})$ bits that are passed on to the next stage, taking as input the $n^{a_i+o(1)}$ bit string from the quantifier and the $n^{b_{i-1}+o(1)})$-bit input from the previous stage of computation. The input to the first stage ($i=1$) is the $n$-bit input string itself. The verifier at the end receives an $n^{b_k+o(1)}$-bit input and uses a $n^{o(1)}$-space machine and $n^{d+o(1)}$ time to compute the final answer. The criteria for acceptance and rejection are analogous to those for $\Sigma_k\P$ and $\Pi_k \P$. 
\end{definition}

Note that our notation obscures $n^{o(1)}$ factors everywhere, although we may occasionally write out small factors for clarity. We index our $b_i$ differently from the notation of~\cite{williams-aut}, as our $b_i$ is their $b_{i+1}$.

\begin{definition}
    \label{def:verif-runtime}
    Given an alternating complexity class $(Q_1 n^{a_1})^{b_1} \dots (Q_k n^{a_k})^{b_k} \TS[n^d]$, we will refer to $n^d$ as the \textbf{verifier runtime}. 
\end{definition}

\begin{lemma}[Speedup Lemma \cite{nepomnjascii}\cite{kannan}]
    \label{lem:speedup}
    For every $0 < x < d$, 
    \begin{equation*}
        \TS[n^d] \subseteq (Q n^x)^{\max(x,1)}(\neg Q x\log n)^1\TS[n^{d-x}].
    \end{equation*}
    Because our notation obscures $n^{o(1)}$ factors, we may write this as
    \begin{equation}\label{eqn:speedup-lemma}
        \TS[n^d] \subseteq (Q n^x)^{\max(x,1)}(\neg Q n^0)^1\TS[n^{d-x}].
    \end{equation}
\end{lemma}

\begin{proof}
    We will prove the lemma when $Q = \exists$; the other case follows because $\TS[n^d]$ is closed under complementation. The idea is that we can break up the transcript of a deterministic computation of length $n^d$ into $n^x$ pieces each of length $n^{d-x}$. Let $M$ be an $n^d$ time machine using $n^{o(1)}$ space. On an $n$-bit input, our $\Sigma_2$ machine will:
    \begin{enumerate}
        \item 
        {\bf Existentially} guess $n^x-1$ intermediate machine configurations $X_1, \dots, X_{n^x-1}$ of $M$, each of size $n^{o(1)}$. These are passed, along with the input, to the next stage. This corresponds to the $(Q n^x)^{\max(x,1)}$ part of the class in \eqref{eqn:speedup-lemma}.
        \item 
        {\bf Universally} quantify over all intermediate configurations, picking one. There are $n^x$ pieces so our quantifier only needs $O(\log n^x) \leq \tilde{O}(n^0)$ bits. If the quantifier picks the $i^{\t{th}}$ configuration, then we pass the state pair $(X_{i-1}, X_i)$ (along with the input) on to the next stage. We take $X_0$ to be the initial configuration of $M$, and $X_n$ to be the (WLOG unique) accepting machine configuration. This corresponds to the $(\lnot Q n^0)^1$ part in \eqref{eqn:speedup-lemma}.
        \item 
        Given input $x$ and a pair of configurations $(X,Y)$ of $M$, the verifier simulates $M$ starting at $X$ for $n^{d-x}$ steps, accepting if the configuration at the end is $Y$ and rejecting otherwise. This corresponds to the $\TS[n^{d-x}]$ part of \eqref{eqn:speedup-lemma}.\qedhere
    \end{enumerate}
\end{proof}

As an extension, we may derive the speedup rule that we will use throughout this paper. 
\begin{corollary}[``Usual'' Speedup Rule, \cite{williams-aut}]
    \label{cor:speedup-rule}
    For every $0 < x < d$, 
    \begin{align*}
        \MoveEqLeft  (Q_1 n^{a_1})^{b_1} \dots (Q_k n^{a_k})^{b_k} \TS[n^d] \\
        &\subseteq (Q_1 n^{a_1})^{b_1} \dots (Q_k n^{\max(a_k,x)})^{\max(b_k,x)} (Q_{k+1} x\log n)^{b_k} \TS[n^{d-x}].
    \end{align*}
\end{corollary}

\begin{proof}
    Observe that we may merge together two quantifiers of the same type. Thus, taking $Q = Q_k$ in \cref{lem:speedup}, we find that
    \begin{align*}
    \MoveEqLeft
        (Q_1 n^{a_1})^{b_1} \dots (Q_k n^{a_k})^{b_k} \TS[n^d]\\ & \subseteq (Q_1 n^{a_1})^{b_1} \dots (Q_k n^{a_k})^{b_k} (Q_k n^x)^{\max(b_k,x)} (Q_{k+1} x\log n)^{b_k} \TS[n^{d-x}] \\
        & \subseteq (Q_1 n^{a_1})^{b_1}\dots (Q_k n^{\max(a_k,x)})^{\max(b_k,x)} (Q_{k+1} x\log n)^{b_k} \TS[n^{d-x}]
    \end{align*}
    where the second containment follows from~\cref{lem:speedup}.
\end{proof}

One might wonder whether we can do any better by also considering the containment arising from taking $Q = \neg Q_k$ in \cref{lem:speedup}. It turns out that any alternation-trading proof using this latter rule can be obtained with \cref{cor:speedup-rule}, and therefore we may safely ignore this option. This is Lemma 3.2 of \cite{williams-aut}. 

\begin{definition}
    We refer to the value $x$ in an application of the speedup rule as the \textbf{speedup parameter} for that application. 
\end{definition}

In \cref{sec:bpts}, we will work with alternating complexity classes with randomized space-bounded verifiers, rather than deterministic ones. We will use a speedup rule due to Diehl and van Melkebeek~\cite{dvm}.  

\begin{theorem}[\cite{dvm}]
    \label{thm:rand-speedup}
    $\BPTS[n^d] \subseteq (\forall n^0)^1 (\exists n^x)^{\max(x,1)}(\forall x\log n) \TS[n^{d-x}]$. 
\end{theorem}

Note again that our notation allows us to hide $n^{o(1)}$ factors. We chose not to obscure the last $x\log n$ bits, as they will be extremely relevant later when we use Grover's algorithm to remove $O(\log n)$-bit quantifiers.

\begin{proof}[Sketch]
    We have
     \begin{align*}
        \BPTS[n^d] & \subseteq (\forall n^0)^1 (\exists n^0)^1 \TS[n^{d}] \\
        & \subseteq (\forall n^0)^1 (\exists n^x)^{\max(x,1)}(\forall x\log n) \TS[n^{d-x}]. \\
    \end{align*}

    The first line here uses Nisan's derandomization of small space machines \cite{nisan} followed by the Sipser-G\'{a}cs-Lautemann theorem. The overall space usage in this step increases by a factor of $\log t(n)$ when applied to a $t(n)$ time machine, but so long as $t(n)$ is polynomial our final machine remains $n^{o(1)}$ space. The second line is an application of~\cref{lem:speedup}, which we can now do because our verifier is deterministic.
\end{proof}

Note again that our notation allows us to obscure $n^{o(1)}$ factors. We choose not to obscure the last $x\log n$ bits, because they will be relevant later when we use Grover's algorithm to remove quantifiers over $O(\log n)$ bits. 

As before, we may express~\cref{thm:rand-speedup} as a rule that can be applied to alternating complexity classes. 

\begin{corollary}[The ``Randomized'' Speedup Rule]
    \label{cor:rand-speedup-rule}
    For every $0 < x < d$,
    \begin{align*}
    \MoveEqLeft
        (Q_1 n^{a_1})^{b_1}\dots (Q_k n^{a_k})^{b_k} \BPTS[n^d] \\ &\subseteq (Q_1 n^{a_1})^{b_1} \dots (Q_k n^{a_k})^{b_k}(Q_{k+1} n^x)^{\max(b_k,x)} (Q_{k+2} x\log n)^{b_k} \TS[n^{d-x}].
    \end{align*}
\end{corollary}

Note that unlike \cref{cor:speedup-rule}, which adds two quantifiers to speed up a deterministic computation, \cref{cor:rand-speedup-rule} adds three\footnote{In both cases, the first quantifier is absorbed into the previous quantifier if one exists, in which case the number of "new" quantifiers is one and two respectively.}.

We've already introduced the usual slowdown lemma, which we restate for convenience. 

\usualslowdown*

The slowdown rule follows from a padding argument and the observation that $\TS[n^c]$ is closed under complementation. 

\begin{corollary}[``Usual'' Slowdown Rule~\cite{williams-aut}]
    Assuming $\NTIME[n] \subseteq \TS[n^c]$, we have 
    \label{cor:slowdownrule}
    \begin{align*}
\MoveEqLeft        (Q_1 n^{a_1})^{b_1} \dots (Q_k n^{a_k})^{b_k} \TS[n^d]\\ &\subseteq (Q_1 n^{a_1})^{b_1} \dots (Q_{k-1} n^{a_{k-1}})^{b_{k-1}} \TS[n^{c\cdot\max(d,b_k,a_k,b_{k-1})}].
    \end{align*}
\end{corollary}

Note that the exponent $b_{k-1}$ is present in the maximum, because our assumption is $\NTIME[n] \subseteq \TS[n^c]$ rather than $\NTIME[n^{\delta}] \subseteq \TS[n^{c\delta}]$ for all $\delta > 0$. 

\begin{definition}
    \label{def:gen-slowdown}
    Let $c, \alpha \in {\mathbb R}$ such that $0 < \alpha \leq 1 < c$. A \textbf{generic slowdown rule with parameters $\alpha$ and $c$} shows that
    \begin{equation*}
        (Q_1 n^{a_1})^{b_1} \dots (Q_k n^{a_k})^{b_k} \TS[n^d] \subseteq (Q_1 n^{a_1})^{b_1} \dots  (Q_{k-1} n^{a_{k-1}})^{b_{k-1}} \TS[n^{c\cdot\max(\alpha d,b_k,a_k,b_{k-1})}].
    \end{equation*}
\end{definition}
Intuitively, having a generic slowdown rule with parameters $c$ and $\alpha$ means that we can turn classes like \[\exists \forall \dots \forall \exists \TIME[n^d]\] into \[\exists \forall \dots \forall \TIME[n^{\alpha c d}].\] We are now ready to define alternation-trading proofs. 
\begin{definition}[\hspace{1sp}\cite{williams-aut}]
    \label{def:alt-proofs}
    An \textbf{alternation-trading proof} is a list of alternating complexity classes, where each subsequent class in the list is contained in the previous class. Each class is derived from the previous by applying one of the following rules. 
    \begin{enumerate}
        \item  If the class is $\TS[n^d]$ (i.e., the verifier is deterministic and there are no quantifiers), we may apply \cref{lem:speedup}: 
            \begin{equation*}
                \TS[n^d] \subseteq (\exists n^x)^{\max(x,1)} (\forall x\log n)^1 \TS[n^{d-x}]
            \end{equation*}
            for some $x \in (0,d)$. 
       \item If the class has at least one quantifier and the verifier is deterministic (i.e., the class ends with $\TS[n^d]$), we may apply \cref{cor:speedup-rule}:
            \begin{align*}
            \MoveEqLeft
                (Q_1 n^{a_1})^{b_1} \dots (Q_k n^{a_k})^{b_k} \TS[n^d] \\ 
                &\subseteq (Q_1 n^{a_1})^{b_1} \dots (Q_k n^{\max(a_k,x)})^{\max(b_k,x)} (Q_{k+1} x\log n)^{b_k} \TS[n^{d-x}]
            \end{align*}
            for some $x \in (0,d)$.
         \item  If the class is $\BPTS[n^d]$ (i.e., the verifier is randomized and there are no quantifiers), we may apply \cref{thm:rand-speedup}: 
            \begin{equation*}
                \BPTS[n^d] \subseteq (\exists n^0)^1 (\forall n^x)^{\max(x,1)}(\exists x\log n) \TS[n^{d-x}] \\
            \end{equation*}
            for some $x \in (0,d)$.
       \item If the class has at least one quantifier and the verifier is randomized (i.e., the class ends with $\BPTS[n^d]$), we may apply \cref{cor:rand-speedup-rule}:
            \begin{align*}
            \MoveEqLeft  (Q_1 n^{a_1})^{b_1}\dots (Q_k n^{a_k})^{b_k} \BPTS[n^d]\\
            &\subseteq (Q_1 n^{a_1})^{b_1} \dots (Q_k n^{a_k})^{b_k}(Q_{k+1} n^x)^{\max(b_k,x)} (Q_{k+2} x\log n)^{b_k} \TS[n^{d-x}]
            \end{align*}
            for some $x \in (0,d)$.
         \item If the class has at least one quantifier, and a generic slowdown rule with parameters $\alpha$ and $c$ hold for the class, we may apply it:
            \begin{align*}
                \dots (Q_k n^{a_k})^{b_k} \mathsf{(BP)TS}[n^d] \subseteq \dots  (Q_{k-1} n^{a_{k-1}})^{b_{k-1}} \mathsf{(BP)TS}[n^{c\cdot\max(\alpha d,b_k,a_k,b_{k-1})}].
            \end{align*}

    \end{enumerate}
    We say that an alternation-trading proof \textbf{shows a contradiction for $c$} if it contains an application of a speedup rule and the proof shows either $\TS[n^d] \subseteq \TS[n^{d'}]$ for $d' \leq d$ or $\BPTS[n^d] \subseteq \BPTS[n^{d'}]$ for $d' \leq d$.
\end{definition}

Note that rules 3 and 4 only apply when proving lower bounds against $\BPTS$.

The containment $\TS[n^d] \subseteq \TS[n^{d'}]$ for $d \geq d'$ does not automatically yield a contradiction\footnote{To apply the naive approach, we need a single machine in $\TS[n^d]$ that can simulate everything in $\TS[n^{d'}]$. However, any fixed machine in $\TS[n^d]$ cannot simulate things use more space than it does. If our simulating machine in $\TS[n^d]$ uses space $f(n) = n^{o(1)}$ then there is always a machine in $\TS[n^{d'}]$ that uses more space while still being $n^{o(1)}$ and our simulating machine cannot simulate this one.}. Fortunately however, we are still able to derive contradictions from this. 

\begin{theorem}[Lemma 3.1 of \cite{williams-aut}]
    If, under the assumption $\NTIME[n] \subseteq \TS[n^c]$, there is an alternation-trading proof with at least two inclusions showing that $\TS[n^d] \subseteq \TS[n^{d'}]$ for $d' \leq d$, then the assumption must have been false and $\NTIME[n] \not \subseteq \TS[n^c]$
\end{theorem}

In \cref{sec:bpts}, we will show that similar statements hold for $\BPTS$ in the contexts in which we need them to hold, thus allowing us to derive contradictions from $\BPTS[n^d] \subseteq \BPTS[n^{d'}]$ for $d' \leq d$. 

\begin{definition}[\cite{williams-aut}]
    An alternating complexity class $(Q_1 n^{a_1})^{b_1} \dots (Q_k n^{a_k})^{b_k} \mathsf{(BP)TS}[n^d]$ is \textbf{orderly} if for all $i \in [k]$, $a_i \leq b_i$.
\end{definition}

\begin{lemma}[\cite{williams-aut}]
    \label{lem:orderly}
    Any alternation-trading proof using rules from \cref{def:alt-proofs} is orderly.
\end{lemma}

As noted by Buss and Williams~\cite{williams-buss}, Lemma~\ref{lem:orderly} implies that, when describing an alternation-trading proof consisting of speedups and slowdowns, it is sufficient to only specify the $\{b_i\}$ and disregard the $\{a_i\}$. We will be somewhat more careful when we apply Grover's algorithm in the quantum setting, but when we can we will simplify our notation by writing only a single exponent inside the parenthesis. Thus, we may abuse notation to write alternating complexity classes in the form 
    $(Q_1 n^{a_1}) \dots (Q_k n^{a_k}) \mathsf{(BP)TS}[n^d],$
    where the $a_i$ are then understood to be the maxima of the corresponding pairs of exponents in the full notation.

As noted by~\cite{williams-buss}, this implies that when describing an alternation-trading proof consisting of speedups and generic slowdowns, it is sufficient to only specify the $\{b_i\}$ and disregard the $\{a_i\}$. We will be somewhat more careful when we apply Grover's algorithm in the quantum setting, but when we can we will simplify our notation by omitting the second exponent.

\begin{definition}[\hspace{1sp}\cite{williams-aut}]
    A \textbf{proof annotation} is a way of specifying a sequence of applications of speedup and slowdown rules. We write $\boldsymbol{1}$ to denote a speedup rule and $\boldsymbol{0}$ to denote a (possibly generic) slowdown rule. When appropriate, we put a subscript under a $\boldsymbol{1}$ to denote the speedup parameter used for that speedup rule application. If $\boldsymbol{s}$ is a proof annotation, we will sometimes write $\boldsymbol{(s)^*}$ to mean an arbitrary number of applications of $\boldsymbol{s}$. Similarly, we will write $\boldsymbol{(s)^{+}}$ to mean an arbitrary but nonzero number of applications of $\boldsymbol{s}$. 
\end{definition}

In this paper, we will work with alternation-trading proofs which apply only one slowdown rule (many times). For such proofs, the sequence of speedups and slowdowns fully determines whether the verifier is deterministic or randomized at a given line of the proof. This means that, when specifying a proof annotation, we do not need to specify which speedup rule we are applying between \cref{cor:speedup-rule} and \cref{cor:rand-speedup-rule}. When the verifier is randomized we must apply \cref{cor:rand-speedup-rule}, and when the verifier is deterministic we should always apply \cref{cor:speedup-rule} as it is strictly more efficient than \cref{cor:rand-speedup-rule}. 

\begin{example}
    In the earlier proof of the Lipton-Viglas bound in \cref{thm:lip-vig}, we used the annotation $\boldsymbol{001}_{c^2}$: two slowdowns and one speedup with speedup parameter $c^2$.
\end{example}

The difficulty in constructing good alternation-trading proofs is applying the speedup and slowdown rules in a good order and choosing the parameters for the speedup rules appropriately. \cite{williams-aut} introduced the use of computer-aided methods to search the proof space and understand the structure of optimal proofs. In particular, \cite{williams-aut} reduced the problem of selecting speedup parameters for a given annotation to a linear programming problem.  

\subsubsection{Annotation Graphs}
It will be useful in \cref{sec:gen-lb} and Appendix~\ref{sec:limits} to visualize proof annotations graphically. 

\begin{definition}
    \label{def:annotation-graph}
    The \textbf{annotation graph} of a proof annotation is a graph with points $(t,h(t))$, where $h(t)$, the \textbf{height}, denotes the quantifier depth of the proof annotation after the first $t$ steps have been applied.
\end{definition}

\cref{fig:annotation-graph} has an example of an annotation graph. Note that the first segment is twice as steep as the other ascending segments because the first speedup rule adds two quantifiers rather than one. 

\begin{figure}[H]
    \centering
    \begin{tikzpicture}
        \begin{axis}[
            xlabel={Step},
            ylabel={Height},
            xmin=-1, xmax=26,
            ymin=-1, ymax=7,
            ticks=none,
            unit vector ratio = 1 1
        ]

        \addplot[
            color=black,
            mark=*,
            ]
            coordinates {
            (0,0)(1,2)(2,3)(3,4)(4,5)(5,4)(6,3)(7,4)(8,5)(9,6)(10,5)(11,4)(12,5)(13,4)(14,5)(15,4)(16,5)(17,4)(18,3)(19,2)(20,3)(21,4)(22,3)(23,2)(24,1)(25,0)
            };
        \end{axis}
    \end{tikzpicture}
    \caption{An example annotation graph, for the annotation $\boldsymbol{1111001110010101000110000}$.}
    \label{fig:annotation-graph}
\end{figure}
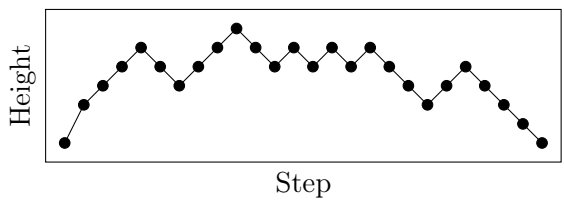

We will sometimes refer to Dyck paths in order to simplify certain definitions. 
\begin{definition}
    \label{def:dyck-path}
    A \textbf{Dyck Path} is a series of up and down steps that never goes below where it started. 
\end{definition}

Note that proof annotations are (almost) in bijection with Dyck paths. In \cite{williams-aut}, this correspondence was used to to programmatically enumerate the annotations.

\subsection{Computational models}

All functions used to bound runtime or space are assumed to be constructible in the given resources. Our model for classical computation is the space-bounded random-access machine, unless specified otherwise. Our proofs are robust to all notions of random access we know. 

Our model for quantum computation will be that of Van Melkebeek and Watson~\cite{vmw}, but our results hold for any reasonable quantum model capable of \emph{obliviously} applying unitaries from a fixed universal set and with quantum random-access to the input.\footnote{Here, ``obliviously'' means that the unitaries applied depend only on the length of the input.}  Recall that if $x \in \{0,1\}^n$ is an input, an algorithm is said to have quantum random-access to $x$ if it can perform the transformation 
\begin{equation*}
    \sum_{i \in [n]} \alpha_i \ket{i}\ket{b} \mapsto \sum_{i \in [n]} \alpha_i \ket{i}\ket{b \oplus x_i},
\end{equation*}
where $i$ is an index and $x_i$ denotes the bit at the $i^{\t{th}}$ position of $x$. The model of \cite{vmw} is capable of simulating all the usual models of quantum computing, and deals carefully with issues like intermediate measurements and numerical precision. 

\subsection{Some Atypical Complexity Classes}

We stated our lower bounds in \cref{sec:intro} in terms of $\QCMA$ and $\MA$. However, our results actually hold for slightly smaller classes, which we describe below. 

\begin{definition}
    \label{def:ebqp}
    The complexity class $\exists \cdot \BQP$ is the set of languages $L$ for which there exists a $\mathsf{BQP}$ verifier $\cA$ such that
    \begin{itemize}
        \item $x \in L \implies (\exists w) \Pr[\cA(x,w) = 1] \geq \f{2}{3}$
        \item $x \not \in L \implies (\forall w) \Pr[\cA(x,w) = 1] \leq \f{1}{3}$.
    \end{itemize}
    We will write $\exists \cdot \BQP_{s,c}$ to denote $\exists \cdot \BQP$ where Arthur has completeness $c$ and soundness $s$. We will write $\exists \cdot \BQTIME[t(n)]$ to denote $\EBQP$ where the length of Merlin's proof and the runtime of the verifier are both $O(t(n))$. 
\end{definition}

We may define $\exists \cdot \BPP$ and $\exists \cdot \BPTIME$ analogously.

\begin{remark}
Note that, while $\exists \cdot \BQP \subseteq \QCMA \subseteq \QMA$ (respectively, $\exists \cdot \BPP \subseteq \MA$), it is not clear if $\exists \cdot \BQP = \QCMA$ ($\exists \cdot \BPP = \MA$) due to differences in the promise conditions. In $\exists \cdot \BQP$, we require that the verifier $\cA(x,w)$ lie in $\BQP$, meaning that it satisfies the promise on every input pair $(x,w)$. On the other hand, in the ``yes'' case of $\QCMA$ (when a string $x$ is in the language), we require only that there exists a polynomial-sized witness $y$ making the verifier $\cA(x,w)$ accept with probability exceeding $\f{2}{3}$. This does not preclude the existence of a witness string $w'$ such that $\f{1}{3} < \Pr[\cA(x,w) = 1] < \f{2}{3}$. 
\end{remark}

\section{Lower Bounds With Generic Slowdowns}
\label{sec:gen-lb}

We will start by introducing a method to reduce the verifier runtime of a class without increasing any of the quantifier exponents, assuming that the verifier runtime isn't too large. This was the main feature in the alternation-trading proofs of \cite{williams-old} that allowed improvement beyond the results of Lipton-Fortnow-van Melkebeek-Viglas~\cite{lvfvm}. 

\begin{lemma}[Generalizes \cite{williams-old}]
    \label{lem:rule-2}
    Let $0 < \alpha \leq 1$ be a real number. Let $c$ be a positive real such that $c < \f{1+\alpha}{\alpha}$. Given any class 
    \[\dots (Q_k n^{a_k}) \TS[n^{a_{k+1}}]\] 
    where $ca_k \leq a_{k+1} < \f{\alpha c}{\alpha c-1}a_k$, there is a nonnegative integer $N := N(a_k)$ such that the annotation $\boldsymbol{(10)^N0}$ with the appropriate speedup parameters proves that \begin{equation*}
        \dots (Q_k n^{a_k}) \TS[n^{a_{k+1}}] \subseteq \dots (Q_k n^{a_k}) \TS[n^{ca_k}] \subseteq \dots  \TS[n^{c^2a_k}]. 
    \end{equation*}
\end{lemma}

\begin{remark}
    \leavevmode
    \begin{itemize}
        \item The use of $a_k$ as the speedup parameter in \cref{lem:rule-2} may seem arbitrary, but it will turn out that this is in fact the best speedup parameter in this setting.
        \item  The statement of the lemma doesn't make sense without the condition $c < \f{\alpha+1}{\alpha}$, since 
            \begin{equation*}
                c < \f{\alpha+1}{\alpha} \iff \alpha c^2 < (\alpha+1)c \iff (\alpha c - 1)c < \alpha c \iff c < \f{\alpha c}{\alpha c - 1}.
            \end{equation*}
        \item Note that if $d\geq ca_k$ then $ca_k$ is the \textit{smallest} that the verifier runtime can ever possibly be after a Dyck path which starts at a clause ending in $(Q_k n^{a_k})$. This is because any Dyck path will conclude with a slowdown rule and $ca_k$ is one of the terms in the maximum. 
    \end{itemize}
\end{remark}

The following proof will follow an argument in~\cite{williams-buss}.
\begin{proof}[Proof of~\cref{lem:rule-2}]

     We will show that finitely many $\boldsymbol{(10)}$ operations with speedup parameter $a_k$ prove that $$\dots (Q_{k-1} n^{a_{k-1}}) (Q_k n^{a_k}) \TS[n^{d}] \subseteq \dots (Q_{k-1} n^{a_{k-1}}) (Q_k n^{a_k}) \TS[n^{ca_k}].$$ Then, a single slowdown step yields the second inclusion. Let's formalize this.

      For brevity, we will write the inclusion we want to show as $a_{k+1} \mapsto ca_k$. If $d = ca_k$ then we take $N = 0$. Otherwise, by the lemma conditions we have $d > ca_k$. A single application of $(1_{a_k}0)$ will send $d \mapsto \min(c\alpha(d - a_k),ca_k)$.  Since we are assuming $d > ca_k$, this is only an improvement --- that is to say, the exponent only decreases --- if $d < \f{\alpha c}{\alpha c - 1}a_k$, which holds by the lemma conditions. We want to show that we can achieve $d \mapsto a_k$ in finitely many steps. After each application of $(1_{a_k}0)$, the verifier runtime is decremented by $$d - c\alpha(d-a_k) = c\alpha a_k - (c\alpha-1)d = \Omega(a_k),$$ where the last step follows by applying the upper bound on $d$ given by the lemma conditions. This lower bound on the decrement is independent of $d$ and thus in finitely many steps we can turn any $d$ satisfying the lemma conditions into $ca_k$. 
\end{proof}

 Based on this lemma, we can define a new rule that we may use in alternation-trading proofs. 

\begin{definition}
    \label{def:rule-2}
    Consider an alternating complexity class $(Q_1 n^{a_1}) \dots (Q_k n^{a_k}) \TS[n^d]$. Given $0 < \alpha \leq 1$ and $c > 0$ satisfying $c < \f{1+\alpha}{\alpha}$, we define the \textbf{squiggle\footnote{To understand this name, consider what an application of this rule looks like in an annotation graph.} rule} with parameters $\alpha, c$ to be the following:
    \begin{itemize}
        \item If $d < \f{\alpha c}{\alpha c - 1}a_k$, apply $\boldsymbol{(1_{a_k}0)^t}$ for $t := \ceil{\f{a_{t+1}}{a_t}}$. 
        \item Otherwise, do nothing. 
    \end{itemize}
    We call an application of the squiggle rule \textbf{proper} if we are in the first case and \textbf{improper} otherwise. {\bf In a proof annotation, when $\alpha$ and $c$ are fixed, we will denote an application of the squiggle rule by $\boldsymbol{2}$.}
\end{definition}

It is worth remembering the ratio \[\f{\alpha c}{\alpha c -1}\] as it will show up frequently in the remainder of this paper. Following the notation of \cref{def:rule-2}, if the value of $d$ (the verifier runtime exponent) for an alternating complexity class is at most $\f{\alpha c}{\alpha c -1}$ times $a_k$ (the exponent in the final quantifier), we may reduce the verifier runtime to its smallest possible value $ca_k$ with an application of \cref{def:rule-2} and \cref{lem:rule-2}, as the following corollary shows.

\begin{corollary}[Corollary of \cref{lem:rule-2}]
    Consider a class $$(Q_1 n^{a_1}) \dots (Q_k n^{a_k}) \TS[n^d].$$ Given $0 < \alpha \leq 1$ and $c > 0$ satisfying $c < \f{1+\alpha}{\alpha}$, applying the squiggle rule (\cref{def:rule-2}) yields the class:
    \begin{itemize}
        \item $(Q_1 n^{a_1}) \dots (Q_k n^{a_k}) \TS[n^{ca_k}]$, if $d < \f{\alpha c}{\alpha c - 1}a_k$.
        \item $(Q_1 n^{a_1}) \dots (Q_k n^{a_k}) \TS[n^d]$ otherwise.
    \end{itemize}
\end{corollary}

We are now in a position to prove \cref{thm:gen-lb}, which we restate for convenience. 

\genlb*

What we actually get is a contradiction for the somewhat-less-natural $$\max\paren{\f{1+\sqrt{1+4\alpha}}{2\alpha}, r_2} < c < \min\paren{\f{1+\alpha}{\alpha}, r_1},$$ where $r_2$ is the second largest root of $P_{\alpha}(x)$. Fortunately, the following lemma lets us clean this up.
\begin{lemma}
    \label{lem:gen-slowdown-params}
    If $r_1$ and $r_2$ are the two largest roots of $P_{\alpha}(x) := \alpha^2x^3 - \alpha x^2 -2\alpha x +1$ then 
    \begin{equation*}
        r_2(\alpha) < \f{1+\sqrt{1+4\alpha}}{2\alpha} < r_1(\alpha) < \f{1+\alpha}{\alpha}
    \end{equation*}
    for all $\alpha > 0$. 
\end{lemma}

We will prove this lemma after proving the main theorem. 

\begin{proof}[Proof of \cref{thm:gen-lb}]
    The lower bound will follow from applying the annotation $\boldsymbol{1^k0(20)^k}$ as $k \rightarrow \infty$. (Recall that {\bf 2} denotes an application of the squiggle rule of \cref{def:rule-2}.) We will choose speedup parameters $\{x_i\}$ so that the following sequence of inclusions is valid and every application of the squiggle rule is proper.  

    \begin{align*}
        \TS[n^d] & \subseteq (\exists n^{x_1}) (\forall n^{x_2}) \dots (Q_k n^{x_k}) (Q_{k+1} n^{x_k}) \TS[n^{(d-x_1-x_2- \dots - x_k)}] & \boldsymbol{1^k}0(20)^k \\
        & \subseteq (\exists n^{x_1}) (\forall n^{x_2}) \dots (Q_k n^{x_k}) \TS[n^{\alpha c(d-x_1-x_2- \dots - x_k)}] & 1^k\boldsymbol{0}(20)^k \\
        & \subseteq (\exists n^{x_1}) (\forall n^{x_2}) \dots (Q_k n^{x_k}) \TS[n^{cx_k}] & 1^k0\boldsymbol{2}0(20)^{k-1} \\
        & \subseteq (\exists n^{x_1}) (\forall n^{x_2}) \dots (Q_{k-1} n^{x_{k-1}}) \TS[n^{\alpha c^2x_k}] & 1^k02\boldsymbol{0}(20)^{k-1} \\ 
        & \subseteq (\exists n^{x_1}) (\forall n^{x_2}) \dots (Q_{k-1} n^{x_{k-1}}) \TS[n^{cx_{k-1}}] & 1^k020\boldsymbol{2}0(20)^{k-2} \\
        & \dots & 1^k02020\boldsymbol{(20)^{k-2}} \\
        & \subseteq \TS[n^{\alpha c^2x_1}].  
    \end{align*}
    In order to ensure a contradiction at the end, we will set $x_1 := \f{d}{\alpha c^2}$. In order to ensure that the first application of the squiggle rule is proper, we require that the $\{x_i\}$ satisfy
    \begin{equation}
        \label{eq:setone}
        \alpha c(d-x_1-\dots - x_k) < \f{\alpha c}{\alpha c-1}x_k \iff d - x_1 - x_2 - \dots - x_k < \f{1}{\alpha c-1}x_k.
    \end{equation}
    In order to ensure that we can apply the squiggle rule properly in all other steps, we require that the $\{x_i\}$ satisfy 
    \begin{equation}
        \label{eq:settwo}
        \alpha c^2x_i < \f{\alpha c}{\alpha c-1}x_{i-1}
    \end{equation}
    for all $2 \leq i \leq k$.

    The next claim lets us find valid speedup parameters $\{x_i\}$ for certain values of $c$ depending on the number of iterations $k$ that we allow in our annotation. The expression will look somewhat hairy, but fortunately most of it will disappear in the limit as $k$ becomes large. 

   \begin{claim}
        Pick $\eps > 0$, and let $\tau := \f{1-\eps}{c(\alpha c -1)}$. There are choices of the speedup parameter $x_i$ such that the annotation $\boldsymbol{1^k0(20)^k}$ implies a contradiction for all $c$ satisfying
        \begin{equation}
            \label{eq:setthree}
            \alpha c^2 - \f{\tau^k-1}{\tau-1} < \f{\tau^k}{\alpha c -1}.
        \end{equation}
    \end{claim}

    \begin{proof}[Proof of Claim]
        Let $x_i := \paren{\f{1-\eps}{c(\alpha c-1)}}^{i-1}x_1 = \paren{\f{1-\eps}{c(\alpha c-1)}}^{i-1}\f{d}{c^2}$ for $i \geq 2$. Observe that all the $x_i$ are positive and satisfy the constraints of \eqref{eq:settwo} for all $i$. If we take $d$ to be sufficiently large we can ensure that every exponent in the proof exceeds $1$. Therefore, if we can show that our selection of $\{x_i\}$ satisfies the first constraint, \eqref{eq:setone}, we will have a valid sequence of rules and, by our choice of $x_1$ above, have derived a contradiction.

        Plugging our choice of $\{x_i\}$ into \eqref{eq:setone} yields the constraint
        \begin{equation}
            d - \f{d}{\alpha c^2}\paren{1+ \f{1-\eps}{c(\alpha c)} + \dots + \paren{\f{1-\eps}{c(\alpha c - 1)}}^{k-1}} < \f{1-\eps}{\alpha c -1}\paren{\paren{\f{1-\eps}{c(\alpha c - 1)}}^{k-1} \f{d}{\alpha c^2}}.
        \end{equation}
        This is equivalent to 
        \begin{equation}
            \alpha c^2 - \f{\tau^k-1}{\tau-1} < \f{\tau^k}{\alpha c -1},
        \end{equation}
        as desired. 
    \end{proof}

    The lemma condition $c > \f{1+\sqrt{1+4\alpha}}{2\alpha}$ means that $\paren{c(\alpha c - 1)}^{-1} < 1$. Now, let $\eps := \f{1}{k}$ and allow $k \rightarrow \infty$ in \eqref{eq:setthree}. Since $\tau \rightarrow \paren{c(\alpha c - 1)}^{-1}$ as $k \rightarrow \infty$, we see that \eqref{eq:setthree} becomes
    \begin{equation*}
        \alpha c^2 - \f{c(\alpha c - 1)}{c(\alpha c - 1)-1} < 0 \iff \alpha^2c^3-\alpha c^2-2\alpha c + 1 < 0
    \end{equation*}
    in the limit. More formally, we have shown that for every $c$ satisfying this constraint, if we take $k$ to be large enough, then our choice of $\{x_i\}$ satisfies \eqref{eq:setone}. The leading coefficient is positive so we satisfy all the constraints (i.e.\ the alternation-trading proof shows a contradiction)  when $c$ lies between the largest and second largest roots of this cubic. Implicitly, we require $c < \f{1+\alpha}{\alpha}$ so that we can apply \cref{lem:rule-2} and $c > \f{1+\sqrt{1+4\alpha}}{2\alpha}$ so that the undesired terms in \eqref{eq:setthree} vanish. By \cref{lem:gen-slowdown-params}, these are satisfied when $\f{1+\sqrt{1+4\alpha}}{2\alpha} < c < r_1$. Furthermore, if $\cC \not \subseteq \TS[n^{c'}]$ for some $c'$ then we automatically have $\cC \not \subseteq \TS[n^c]$ for all $c < c'$.
\end{proof}

\begin{proof}[Proof of \cref{lem:gen-slowdown-params}]
    We need to prove the inequalities
     \begin{equation}
        \label{eq:no1}
        r_2(\alpha) < \f{1+\sqrt{1+4\alpha}}{2\alpha} < r_1(\alpha)
    \end{equation}
    and
    \begin{equation}
        \label{eq:no2}
        r_1(\alpha) < \f{1+\alpha}{\alpha}.
    \end{equation}

    To make the ensuing description more clear, let $t(\alpha) :=\f{1+\sqrt{1+4\alpha}}{2\alpha}$, and let $r_3(\alpha)$ be the third root of $P_{\alpha}(x)$. Let's start with \eqref{eq:no1}. We will plug $t(\alpha)$ into $P(x)$ and show that it is negative for all $\alpha$.

    Let $\beta := \sqrt{1+4\alpha}$. Then, 
    \begin{equation*}
         P\paren{\f{1+\beta}{2\alpha}} = \f{(1+\beta)^3}{8\alpha} - \f{(1+\beta)^2}{4\alpha} - \beta.
    \end{equation*}

    Clearing denominators and expanding yields
    \begin{equation*}
        \beta^3-\beta^2-\beta-1-8\alpha\beta = -2-4\alpha(1+\sqrt{1+4\alpha}) < 0.
    \end{equation*}

    This shows that for all $\alpha > 0$, $P_{\alpha}(t_1(\alpha)) < 0$, meaning that either $r_3(\alpha) > t(\alpha)$ or $r_2(\alpha) < t(\alpha) < r_1(\alpha)$. However, by continuity of everything involved with respect to $\alpha$ (continuity of the roots with respect to $\alpha$ follows from Vieta's formulas), we see that only one of these can be true across $\alpha > 0$ if $P(t(\alpha)) \neq 0$. Therefore, it's sufficient to check at any single $\alpha$ whether or not $t(\alpha) > r_2(\alpha)$. Observe that $t\paren{\f{1}{2}} = 1+\sqrt{3} > r_2\paren{\f{1}{2}} \approx 0.806$, and hence for all $\alpha > 0$, $r_2(\alpha) < t(\alpha) < r_1(\alpha)$.

    We can do something similar for the second inequality, except here we will show that the result is always positive. 
    \begin{align*}
        P\paren{\f{1+\alpha}{\alpha}} & = \alpha^2\paren{\f{1+\alpha}{\alpha}}^3-\alpha\paren{\f{1+\alpha}{\alpha}}^2-2\alpha\paren{\f{1+\alpha}{\alpha}}+1 \\
        & = \f{(1+\alpha)^3}{\alpha} - \f{(1+\alpha)}{\alpha} - 2(1+\alpha)+1.
    \end{align*}

    Because $\alpha > 0$, we can clear denominators without changing the sign
    \begin{align*}
        & = (1+\alpha)^3-(1+\alpha)^2-2\alpha^2-\alpha \\
        & = \alpha^3 > 0.
    \end{align*}

    This shows that for all $\alpha > 0$, $P_{\alpha}(t_1(\alpha)) > 0$, meaning that either $r_3(\alpha) < t_1(\alpha) < r_2(\alpha)$ or $r_1(\alpha) < t_1(\alpha)$. Testing at $\alpha = 1$, following the logic above, is sufficient to derive the result.  
\end{proof}

From \cref{thm:gen-lb}, we immediately obtain the following corollary. 
    
\begin{corollary}[\hspace{1sp}\cite{williams-mod}]
    For $c < 2\cos(\f{\pi}{7}) \approx 1.801$, $\NTIME[n] \not \subseteq \TS[n^c]$.
\end{corollary}
\begin{proof}
    Taking $\alpha = 1$ (as this is normal slowdown) yields $c^3-c^2-2c+1 < 0$. The largest root is $2\cos(\f{\pi}{7})$, so we have $\NTIME[n] \not \subseteq \TS[n^{2\cos(\f{\pi}{7})-o(1)}]$
\end{proof}

\section{A Generic Slowdown Rule From Grover's Algorithm}
\label{sec:grodown}
Now that we've proved the general case, we turn to an application. Observe that any application of a speedup rule with parameter $x$ results in an alternating complexity class whose final quantifier is over $x\log n$ bits, as this quantifier indexes over the $n^x$ states that it receives from the penultimate quantifier. For example, 
\begin{equation*}
    \TS[n^d] \subseteq (\exists n^x) (\forall x \log n)^1 \TS[n^{d-x}]. 
\end{equation*}
Normally, we could remove the last quantifier with a slowdown rule, yielding the inclusion
\begin{equation}
    \label{eq:slowdown-removal}
    \TS[n^d] \subseteq (\exists n^x) \TS[n^{c\cdot \max(x,d-x,1)}].
\end{equation}
We could also remove the quantifier by having the deterministic verifier try all $n^x$ possible strings it could receive, but this yields the useless inclusion $\TS[n^d] \subseteq (\exists n^x) \TS[n^d]$. However, if we allow alternating complexity classes with quantum verifiers rather than deterministic verifiers, we can remove the last quantifier more efficiently. In particular, we can think of the last quantifier as a search problem over a space of size $N := n^x$. Classically, no blackbox algorithm can search over $N$ items with fewer than $N$ queries in the worst case, but in the quantum setting Grover's algorithm lets us do this in $O(\sqrt{N})$ queries! This gives us, for some informal notion of quantum time,
\begin{equation}
    \label{eq:grover-removal}
    \TS[n^d] \subseteq (\exists n^x) \mathsf{QTIME}[n^{d-\f{x}{2}}].
\end{equation}
This method of quantifier removal allows us to remove quantifiers more efficiently than \eqref{eq:slowdown-removal} when $c$ is large at the cost of introducing a quantum verifier. However, the quantum verifier can be replaced with a deterministic verifier, under the assumption $\QCMATIME[n] \subseteq \TS[n^c]$! Ultimately, we will find that combining a speedup (adding one quantifier), \eqref{eq:grover-removal} (removing one quantifier), and the assumption (removing one quantifier) yields a generic slowdown rule with $\alpha = \f{2}{3}$.

\subsection{Review of Grover's Algorithm}

Given (quantum) query access to a function $f \colon [N] \rightarrow \{0,1\}$, Grover's algorithm tells us whether or not there exists an $\alpha \in [N]$ such that $f(\alpha) = 1$. For a given $f$, let $S := \{\alpha \in [N] \colon f(\alpha) = 1\}$. Intuitively, each Grover iteration rotates $\ket{\psi}$ by an $\Omega\paren{\f{1}{\sqrt{N}}}$ angle in the plane spanned by $\ket{u}$ and the ``center of mass'' of $S$.\footnote{We start out with a state symmetric in all coordinates and at each step apply operators that act on all basis vectors of $S$ equally and reduce the support of $\ket{\psi}$ off of $S$.} The success probability is thus periodic with respect to the number of iterations. It turns out, per \cite{grover}, that the probability of success of Grover search after $j$ iterations is $\sin^2(2(j+1)\theta)$ where $\theta = \sin^{-1}\sqrt{\f{|S|}{N}}$. Regardless of $\f{|S|}{N}$, a sufficiently large random number of iterations should succeed with probability roughly $\f{1}{2}$, which is the average value of $\sin^2 x$. The following lemma of \cite{bbht96} formalizes this.

\begin{figure}
    \fbox{%
        \begin{minipage}[c]{\textwidth}
            \begin{enumerate}
                \item \textbf{Initialize} $\ket{\psi}$ to $\ket{u} := \f{1}{\sqrt{N}}\sum_{i \in [N]} \ket{i}$.
                \item \textbf{Grover Iteration}: Apply the following operators to $\ket{\psi}$ $O(\sqrt{N})$ times:
                    \begin{enumerate}
                        \item \textbf{Grover Diffusion}: $I - 2\sum_{\alpha \in S}\ket{\alpha}{\bra{\alpha}}$.
                        \item \textbf{Inversion about the Mean}: $2\ket{u}\bra{u}-I$.
                    \end{enumerate}
                \item \textbf{Measure}: Apply $f$ to $\ket{\psi}$ and measure the output qubit. 
            \end{enumerate}
        \end{minipage}%
    }
    \caption{An overview of Grover's algorithm}
    \label{fig:grover-box}
\end{figure}

    \begin{lemma}[\cite{bbht96}]
        Let $k$ be an arbitrary positive integer and let $j$ be an integer chosen uniformly at random from $[0,k-1]$. If we observe the register after applying $j$ Grover iterations starting from the uniform state, the probability of obtaining a solution is at least $\f{1}{4}$ when $k \geq \f{1}{\sin 2\theta}$.
    \end{lemma}

    \begin{corollary}
        \label{cor:grover-random-iter}
        Let $j$ be an integer chosen uniformly at random from $[0, (\sin \f{2}{\sqrt{N}})^{-1}]$. If we observe the register after applying $j$ Grover iterations starting from the uniform state, the probability of obtaining a solution is at least $\f{1}{4}$. 
    \end{corollary}
    \begin{proof}
        \begin{equation*}
            \theta =\sin^{-1} \sqrt{\f{|S|}{N}} \geq \sin^{-1} \sqrt{\f{1}{N}} \geq \f{1}{\sqrt{N}}
        \end{equation*}
        and plugging this in we have what we want. 
    \end{proof}

    Note that $(\sin \f{2}{\sqrt{N}})^{-1}$ is $O(\sqrt{N})$. The point of all this is that we can simply guess how many iterations to do and still ensure that our success probability is at least $\f{1}{4}$. Repeating this algorithm four times yields a success probability exceeding $\f{2}{3}$.

\subsection{Using Grover's Algorithm to Invert RAMs}
In order to apply Grover's algorithm to remove the quantifier in expressions like $(\exists \log n) \TS[n^d]$ --- specifically, to perform Grover diffusion  --- we must be able to implement the relevant function $f \in \TS[n^d]$ in our quantum computational model. 

For small space machines, sequential computation can be made oblivious\footnote{Roughly speaking, a computation is said to be \textbf{oblivious} if the sequence of steps performed by the computer only depends on the input length. As a model of computation, a circuit is oblivious in that the same circuit is used on every input of a given length. In contrast, Turing machines and RAMs are not oblivious in general, although some might happen to be.} without a significant runtime overhead, and classical oblivious computation can be replaced with a quantum circuit in the usual manner. However, it is not immediately clear how to apply this approach to random-access machines.

We will resolve this issue by replacing random-access calls to the workspace by sequential access (incurring an $O(s)$ blowup) and by using quantum random-access to the input to replace random-access calls to the input. Quantum random-access lets us do several different random-accesses in superposition, so each of the ``parallel'' threads of computation can make independent queries and receive independent information. 

\begin{remark}
    Note that the use of quantum random-access would be critical even if our input was a database containing the values $f(y)$, as we still need to efficiently perform the operation
    \begin{equation*}
        \sum_y \alpha_y\ket{y}\ket{0} \mapsto \sum_y \alpha_y\ket{y}\ket{f(y)}
    \end{equation*}
during Grover diffusion. 
\end{remark}

The following lemma converts the classical random-access machine of $f$ to a ``normal form'' that is more amenable to implementation in a quantum computer. The workspace of $A$ includes an input address register used to specify which bit of the input the machine wishes to read and an input query bit which stores the result of the most recent query of the input. Assume without loss of generality that the first $\ceil{\log n}$ bits of the workspace hold the input address register, and bit $\ceil{\log n} + 1$ holds the current bit of input being read.  
    \begin{lemma}[Normal Form Circuit for Time-Space Bounded Computation]
        \label{lem:ram2nice}
        Let $f$ be a function computed by a random-access machine $A$ in time $t(n)$ and space $s(n)$ on inputs of length $n$. Then, there exists a sequence of uniform Boolean circuits $\cC_1, \dots, \cC_r$, such that:
        \begin{enumerate}
            \item for all $1 < i \leq r-1$, $\cC_i$ has $s$ inputs and $s$ outputs,
            \item $\cC_1$ has no inputs and $s$ outputs,
            \item $\cC_r$ has $s$ inputs and $1$ output,
            \item $\sum_{i=1}^r |\cC_i| = O(ts^2)$ where $|\cC_i|$ denotes the size of $\cC_i$,
        \end{enumerate}
        and furthermore 
        \begin{equation}
            f(y) = \cC_r \circ R \circ \cC_{r-1} \circ \dots \circ R \circ \cC_1,
        \end{equation}
        where $R \colon \{0,1\}^{s(n)} \rightarrow \{0,1\}^{s(n)}$, on input $z \in \{0,1\}^{s(n)}$, sets $z_{\ceil{\log n}+1} = y_{z_1z_2\dots z_{\ceil{\log n}}}$ and leaves the remainder of $z$ unchanged. Here, we write $s_j$ to denote the $j^{\t{th}}$ bit of a string $s$. 
    \end{lemma}

    In \cref{lem:ram2nice}, the circuits $\cC_i$ are used to simulate steps in the workspace of the machine, and the function $R$ perform random accesses to the $n$-bit input $y$, by setting its $(\ceil{\log n}+1)^{\t{th}}$ output bit to the bit of $y$ whose index is specified by the first $\ceil{\log n}$ bits. Note that $y$ is the original length-$n$ input to $f$ and $n$ may be much larger than $s(n)$.
    Intuitively, this lemma makes $A$ ``as oblivious  as possible'', ensuring that the random accesses occur at the same timestamps regardless of the input and ensuring that everything other than the random accesses is a uniform circuit. 

    \begin{proof}
        First, replace all random-access to the workspace in $A$ with sequential access. This increases the runtime of our algorithm to $O(ts)$. Assume without loss of generality that when the input address register is changed, the input query bit is automatically updated in the next step. This is fine because we can always defer changing the input address register until after we have copied the input query bit to another location. In the notation of the lemma, this means that $R$ occurs every other step. The state function of the machine itself can be represented by an $O(s)$ size circuit with $s$ inputs and $s$ outputs in the standard way. Therefore, the $\cC_i$ for $1 \leq i \leq r-1$ can all be implemented in this fashion. $\cC_r$ takes input $z$ and returns only the output bit of the workspace. There are $O(ts)$ steps in the sequential computation, so the total circuit size is $O(ts^2)$.
    \end{proof}

    \begin{corollary}
        \label{cor:ram2quant}
        Let $f: \{0,1\}^m \rightarrow \{0,1\}$ be a function that can be computed in classical time $t$ and space $s$ by a random-access machine. Then, the transformation $\ket{y}\ket{0} \mapsto \ket{y}\ket{f(y)}$ can be implemented by a quantum computer with quantum random-access to its input in time $O(ts^2)$. 
    \end{corollary}

    \begin{proof}
        Let $\{\cC_i\}$ be the circuits arising from applying \cref{lem:ram2nice} to the random access machine computing $f$. We implement $f$ with a uniform quantum circuit as follows. We can replace every $R$ with a QRAM call. Because the QRAM calls occur at fixed timesteps we can uniformly specify when they should be made. Each QRAM call takes $O(1)$ time. We can replace each $\cC_i$ with a uniform quantum circuit, as each gate of the Boolean circuit corresponds to $O(1)$ unitaries. In total our quantum algorithm thus takes time $O(ts^2)$. 
    \end{proof}

    \begin{lemma}[Grover's Algorithm for Time-Space Bounded Computation]
        \label{lem:grover}
        Given a function $f: \{0,1\}^m \rightarrow \{0,1\}$ that can be computed in classical time $t$ and space $s$ by a random-access machine, there is a quantum algorithm taking time $O(2^{\f{m}{2}}(ts^2+m))$ that finds a value $\alpha \in \{0,1\}^m$ such that $f(\alpha) = 1$, assuming one exists, with probability at least $\f{2}{3}$.  
        In particular, when $m = x\log n, t = n^d, s = O(\log n)$, we obtain a time bound of $O(n^{d+\f{x}{2}}+o(1))$.
    \end{lemma}

    \begin{proof}
        We use Grover's algorithm, and follow the steps and notation of~\cref{fig:grover-box}.

        \begin{enumerate}
            \item \textbf{Initialize}: Requires $O(m)$ Hadamard gates to prepare the uniform state on $m$ qubits. 
            \item \textbf{Grover Iteration}: Per \cref{cor:grover-random-iter}, we repeat the following $O(2^{\f{m}{2}})$ times:
                \begin{enumerate}
                    \item \textbf{Grover Diffusion}: Dominated by the cost of implementing the function itself, which by \cref{cor:ram2quant} is $O(ts^2)$. 
                    \item \textbf{Inversion about the Mean}: Apply Hadamards to each qubit of $\ket{\psi}$, apply \[2\ket{00\dots0}\bra{00\dots0} - I,\] and again apply Hadamards. The middle step is equivalent to a $Z$ controlled by the OR of all the qubits in $\ket{\psi}$. Thus, inversion about the mean requires $O(m)$ time. 
                \end{enumerate}
            \item \textbf{Measure}: We need to run the function once and then measure. The implementation, by \cref{cor:ram2quant}, takes $O(ts^2)$ time.        
        \end{enumerate}
    \end{proof}

    \subsection{Proving \texorpdfstring{\cref{thm:qcma-ts-lb}}{Main Theorem 2}}

    We are now in a position to apply Grover's algorithm to prove a generic slowdown rule. 
    \begin{lemma}
        \label{lem:grodown}
        If $\EBQTIME_{\f{1}{3},1}[n] \subseteq \TS[n^c]$ then 
        \begin{equation*}
            \dots (Q_{k-1} n^{a_{k-1}})^{b_{k-1}} (Q_k n^{a_k})^{b_k} \TS[n^{d}] \subseteq \dots (Q_{k-1} n^{a_{k-1}})^{b_{k-1}} \TS[n^{c\cdot \max(a_k, b_k,  b_{k-1}, 1, \f{2d}{3})}].
        \end{equation*}
    \end{lemma}

    Therefore, under the assumption $\EBQTIME_{\f{1}{3},1}[n] \subseteq \TS[n^c]$ we may apply a generic slowdown rule with parameters $\alpha = \f{2}{3}$ and $c$. \cref{lem:grodown} corresponds to the following sequence of operations:

    \begin{enumerate}
        \item Classical Speedup (\cref{cor:speedup-rule})
        \item Grover's algorithm to invert a classical function with classical input (\cref{lem:grover})
        \item Direct application of $\EBQTIME[n] \subseteq \TS[n^c]$ 
    \end{enumerate}

    \begin{proof}
        By classical speedup, we have
        \begin{equation}
            \label{eq:blah1}
            \dots (Q_k n^{a_k})^{b_k} \TS[n^{d}] \subseteq \dots (Q_k n^{\max(a_k,x)})^{\max(b_k,x)} (Q_{k+1} x\log n)^{b_k} \TS[n^{d-x}].
        \end{equation}
        Consider the function $g \colon \{0,1\}^{n^{\max(b_k, x)}} \times \{0,1\}^{x\log n} \rightarrow \{0,1\}$ which implements the $\TS[n^{d-x}]$ verifier from the right-hand-side of \eqref{eq:blah1} given as inputs the $n^{\max(b_k,x)}$ bits of output from the $k^{\t{th}}$ stage/quantifier and the $(x \log n)$-bit string chosen by the $(k+1)^{\t{th}}$ stage/quantifier. We will apply \cref{lem:grover} to the function $g_z := g(z, \cdot)$, where $z$ is the output from the $k^{\t{th}}$ stage of the class (not the $(k+1)^{\t{th}}$ stage!). In particular, we can use Grover's algorithm to search over the space of possible values of the last quantifier of $(x \log n)$ bits. Thus, the inputs to $g_z$ are strings of length $m = x\log n$. The runtime of $g_z$ is $O(ts^2)$ by \cref{cor:ram2quant}, as $g_z$ just needs to evaluate the verifier on $z$ (the output of the $k^{\t{th}}$ stage) and a length $m$-input. Therefore, applying \cref{lem:grover},
        \begin{equation*}
            \subseteq \dots (Q_{k-1} n^{a_{k-1}})^{b_{k-1}} (Q_k n^{\max(a_k,x)})^{\max(b_k,x)} \mathsf{BQTIME}[n^{d-\f{x}{2}}]. 
        \end{equation*}
        Without loss of generality, suppose that $Q_k = \exists$. Then, we can continue the sequence of inclusions as follows:
        \begin{align*}
         \MoveEqLeft  \subseteq \dots (Q_{k-1} n^{a_{k-1}})^{b_{k-1}} \EBQTIME[n^{\max(a_k, b_k,x , d-\f{x}{2})}]\\
            & \subseteq \dots (Q_{k-1} n^{a_{k-1}})^{b_{k-1}} \TS[n^{c\cdot \max(b_{k-1}, a_k, b_k, x, d-\f{x}{2}, 1)}].
        \end{align*}
        (Note we need a $1$ in the maximum in the last equation, because our assumption \linebreak
        ${\EBQTIME_{\f{1}{3},1}[n] \subseteq \TS[n^c]}$ only implies $\EBQTIME[n^{\delta}] \subseteq \TS[n^c]$ for $\delta < 1$.) 
        To minimize the exponent, we take $x = \f{2a_0}{3}$, yielding the exponent in the lemma statement. 
    \end{proof}
    \begin{remark}
        Note that in the proof of \cref{lem:grodown} we do not need to account for the the $n^{b_k}$ exponent on the $(k+1)^{\t{th}}$ stage/quantifier when computing the the runtime of $g_z$ when preparing to apply \cref{lem:grover}. This is because the quantifier $(Q_{k+1} x\log n)^{b_k}$ on the right-hand-side of \eqref{eq:blah1}, which arises due to a speedup rule, has a $b_k$ in the exponent only because it needs to copy and pass on the output of the $k^{\t{th}}$ stage. This is normally necessary because the verifier on the right-hand-side of \eqref{eq:blah1} needs to run from configuration to configuration on the original verifier's input and hence needs to be able to access the original input (i.e., the input to the verifier on the left-hand-side of the inclusion). However, by our definition of $g_z$, we don't need to worry about the cost of copying the output from the $k^{\t{th}}$ stage $z$ is itself the output of the $k^{\t{th}}$ stage. We don't need to copy and pass on the original input because $g_z$ already has it! Put differently, because we're collapsing two stages into one, we don't need to expend time on computation that's only used to pass input along. 
    \end{remark}

    Note that every possible proof involving slowdown and Grover's algorithm (\cref{lem:grover}) can be done using \cref{lem:grodown} as the only rule for removing quantifiers, as we can only apply~\cref{lem:grover} between a speedup and a slowdown. We can replace any such block with \cref{lem:grodown}. Furthermore, any isolated regular $(\alpha = 1)$ slowdown can be replaced by \cref{lem:grodown} without worsening the resulting verifier's running time.  

    Plugging $\alpha = \f{2}{3}$ into \cref{thm:gen-lb}, we obtain the following corollary.

    \begin{corollary}
        \label{cor:ebqp-lb}
        $\exists\cdot \mathsf{BQTIME}_{\f{1}{3}, 1}[n] \not \subseteq \TS[n^c]$ for $c < \f{3+\sqrt{3}}{2} \approx 2.366$.
    \end{corollary}

    \begin{proof}
        Taking $\alpha = \f{2}{3}$ (per \cref{lem:grodown}) yields $4c^3-6c^2-12c+9 = 0$. The largest root is $\f{3+\sqrt{3}}{2}$, so we have $\EBQTIME[n] \not \subseteq \TS[n^{\f{3+\sqrt{3}}{2}-o(1)}]$ by \cref{thm:gen-lb}.
    \end{proof}

    Since $\EBQTIME[n] \subseteq \QCMATIME[n]$, this proves \cref{thm:qcma-ts-lb}.

\section{Lower Bounds Against \texorpdfstring{$\mathsf{BPTS}$}{BPTS}}
\label{sec:bpts}

We turn to proving lower bounds against randomized algorithms. In this section, we will use both \cref{cor:speedup-rule} and \cref{cor:rand-speedup-rule} as speedup rules. The randomized slowdown rule is straightforward. 

\begin{lemma}[``Randomized'' Slowdown Rule]
    \label{lem:randslowdownrule}
    Assuming that $\exists \cdot \BPTIME[n] \subseteq \BPTS[n^c]$, we have
    \begin{equation*}
        (Q_1 n^{a_1})^{b_1} \dots (Q_k n^{a_k})^{b_k} \BPTS[n^d] \subseteq (Q_1 n^{a_1})^{b_1} \dots (Q_{k-1} n^{a_{k-1}})^{b_{k-1}} \BPTS[n^{c\cdot\max(d,b_k,a_k,b_{k-1})}].
    \end{equation*}
\end{lemma}

\begin{lemma}
    \label{lem:rand-contr}
    Suppose that, under the assumption $\exists \cdot \BPTIME[n] \subseteq \BPTS[n^c]$ for some $c$, there is an alternation-trading proof with at least two inclusions proving that $\BPTS[n^d] \subseteq \BPTS[n^{d'}]$ for $d < d'$. Then, the assumption is false. 
\end{lemma}

\begin{proof}
    Suppose for some $d > 0$, we have an alternation-trading proof of the form
    \begin{align*}
        \BPTS[n^d] & \subseteq (\exists n^0)^1 (\forall n^x)^{\max(x,1)}(\exists x^{\eps}) \TS[n^{d-x}] \\
        & \subseteq \dots \\
        & \subseteq \BPTS[n^{d'}] \\
    \end{align*}
    We claim that there exists a $\eps > 0$ such that
    \begin{align*}
        \BPTS[n^d] & \subseteq (\exists n^{\eps})^{1} (\forall n^x)^{\max(x,1)}(\exists x^{\eps}) \TS[n^{d-x}] \\
        & \subseteq \dots \\
        & \subseteq \BPTS[n^{d'}]. \\
    \end{align*}    
    To see this, note that the exponents we have changed are relevant only during a certain (finite) number of slowdown rules. For sufficiently small $\eps > 0$, the maxima in the exponents of all these applications of the slowdown rule do not change, meaning that our final exponent $d'$ is unchanged.
    
    Now, applying a single speedup rule with parameter $x' := x -\f{d-d'}{2}$, we find that 
    \begin{align*}
        \BPTS[n^d] & \subseteq (\exists n^{\eps})^1 (\forall n^x)^{\max(x,1)}(\exists x\log n) \TS[n^{d-x}] \\
        & \subseteq \dots \\
        & \subseteq \BPTS[n^{d'}] \\
        & \subseteq (\forall n^0)^1 (\exists n^{x'})^{\max({x'},1)}(\forall n^0) \TS[n^{d'-x'}].
    \end{align*}  
    By a fine-grained version of the alternating time hierarchy theorem (proved by direct diagonalization), we have
    \begin{equation*}
        (\exists n^{\eps})^1 (\forall n^x)^{\max(x,1)}(\exists x\log n) \TS[n^{d-x}] \not \subseteq (\forall n^0)^1 (\exists n^{x'})^{\max({x'},1)}(\forall n^0) \TS[n^{d'-x'}],
    \end{equation*}
    which yields the desired contradiction.
\end{proof}

Recall that by \cref{lem:orderly}, we may suppress one of the exponents when writing alternating complexity classes. We will do so throughout the remainder of this section.

Now that we have the preliminaries out of the way, let us prove some lower bounds. As before, we will use $1$ to denote a speedup and $0$ to denote a slowdown. As mentioned in \cref{sec:prelim}, there will be no ambiguity as to which needs to be applied at any point in an alternation-trading proof. The choice of randomized vs.\ usual speedup is entirely dependent on the verifier at that step --- if the verifier in the alternating complexity class is probabilistic, we use randomized speedup and otherwise, we use the usual speedup.

The first part of \cref{thm:bpts-lb} is an immediate corollary of the following theorem. 

\begin{theorem}
    $\exists \cdot \BPTIME[n] \not \subseteq \BPTS[n^c]$ for $c < r_1$, where $r_1 \approx 1.466$ is the largest root of the polynomial $x^3-x^2-1 = 0$.
\end{theorem}

\begin{proof}
    Our analysis is similar to the proof of \cref{thm:gen-lb}. The bound will arise by applying the annotation $\boldsymbol{1^k0^{k+2}}$ with the appropriate speedup parameters as $k \rightarrow \infty$. We will choose parameters $\{x_i\}$ so that the following sequence of inclusions is valid.
    \begin{align*}
        \BPTS[n^d] & \subseteq (\exists n^1) (\forall n^{x_1})(\exists n^{x_2}) \dots (Q_k n^{x_k}) (Q_{k+1} n^{x_k}) \TS[n^{(d-x_1-x_2- \dots - x_k})] & \boldsymbol{1^k}0^{k+2} \\
        & \subseteq (\exists n^1) (\forall n^{x_1})(\exists n^{x_2}) \dots (Q_k n^{x_k}) \BPTS[n^{c(d-x_1-x_2- \dots - x_k)}] & 1^k\boldsymbol{0}0^{k+1} \\
        & \subseteq (\exists n^1) (\forall n^{x_1}) \dots (Q_k n^{x_k}) \BPTS[n^{cx_k}] &  \\
        & \subseteq (\exists n^1) (\forall n^{x_1}) \dots (Q_{k-1} n^{x_{k-1}}) \BPTS[n^{c^2x_k}] & 1^k0\boldsymbol{0}0^k \\
        & \subseteq (\exists n^1) (\forall n^{x_1}) \dots (Q_{k-1} n^{x_{k-1}}) \BPTS[n^{cx_{k-1}}] & \\
        & \dots & 1^k00\boldsymbol{0^{k-1}}0 \\
        & \subseteq (\exists n^1) \BPTS[n^{c^2 x_1}] \\
        & \subseteq \BPTS[n^{c^3 x_1}] & 1^k0^{k+1}\boldsymbol{0} \\
        & \subseteq \BPTS[n^{d}] 
    \end{align*}
    In order for all of the above inclusions to hold, we take \[x_1 := \f{d}{c^3},~ ~ x_i := (c(1-\epsilon))^{1-i}\cdot x_1 = \f{d(1-\eps)^{1-i}}{c^{i+2}}\] for $\eps := \f{1}{k}$. This automatically satisfies all the constraints except for the one corresponding to the third line, which requires
    \begin{align*}
        & c(d-x_1-x_2- \dots - x_k) \leq cx_k \\
        \iff & 1 - \f{1}{c^{3}}\paren{\f{1-\eps}{c}}^{k-1} - \f{1}{c^3}\sum_{i = 0}^{k-1} \paren{\f{1-\eps}{c}}^i < 0  \\
        \iff & 1 - \f{1}{c^{3}}\paren{\f{1-\eps}{c}}^{k-1} - \f{1}{c^3}\f{1-\f{1-\eps}{c^k}}{1-\f{1-\eps}{c}} < 0.
    \end{align*}
    As $k \rightarrow \infty$, several terms vanish. We are left with 
    \begin{equation*}
        1 - \f{1}{c^3(1-\f{1}{c})} < 0 \iff c^3 - c^2 - 1 < 0.
    \end{equation*}
    The largest root is at $c \approx 1.466$. 
\end{proof}

When we are allowed to introduce quantum operations and use \cref{lem:grodown}, we can do slightly better than this.

\begin{theorem}
    $\exists \cdot \BQTIME[n] \not \subseteq \BPTS[n^c]$ for $c < r_1$, where $r_1$ is the largest root of the polynomial $x^3-x^2-1 = 0$. 
\end{theorem}

The second part of \cref{thm:bpts-lb} is an immediate corollary of the previous theorem. 

\begin{proof}
    Let $d > 1.5$. We have the following sequence of containments:
    \begin{align*}
        \BPTS[n^d] & \subseteq (\exists n^1)(\forall n^{\f{2d}{3}})(\exists \f{2d}{3} \log n)^1 \TS[n^{\f{d}{3}}] \\
        & \subseteq (\exists n^1)(\forall n^{\f{2d}{3}}) \BQTIME[n^{\f{2d}{3}}] \\
        & \subseteq (\exists n^1)\BPTS[n^{\f{2cd}{3}}], 
    \end{align*}
    where we've used \cref{lem:grover} in the second inclusion. When $c < 1.5$, we have $\f{2c}{3} < 1$. Thus, we can repeat this procedure until we derive the inclusion
    \begin{align*}
        \BPTS[n^d] & \subseteq (\exists n^1)\BPTS[n^{\f{d}{c}}] \\
        & \subseteq \BPTS[n^{d'}]
    \end{align*}
    for $d' < d$.
    To get a contradiction, we can more or less use the ideas of \cref{lem:rand-contr}. The only caveat is that we cannot naively replace the clause $(\exists \f{2d}{3} \log n)^1$ in the first inclusion above with $(\exists n^{\eps})^1$ as we did in \cref{lem:rand-contr} because we need the quantifier to be over $O(\log n)$ bits in order to apply \cref{lem:grover}. Fortunately, we can simply use the normal slowdown rule rather than \cref{lem:grover} the first time and then use \cref{lem:grover} for the rest of the repetitions. More precisely, we can say that 
    \begin{align*}
        \BPTS[n^d] & \subseteq (\exists n^1)(\forall n^{\f{2d}{3}})(\exists \f{2d}{3} \log n)^1 \TS[n^{\f{d}{3}}] \\
        & \subseteq (\exists n^1)(\forall n^{\f{2d}{3}}) \BPTS[n^{\f{2cd}{3}}] \\
        & \subseteq (\exists n^1)\BPTS[n^{\f{2c^2d}{3}}], \\
    \end{align*}
    and we can still repeat as before until $\BPTS[n^d] \subseteq \BPTS[n^{d'}]$ for $d' < d$. Now, we can repeat the same argument as in \cref{lem:rand-contr} to obtain a contradiction. 
\end{proof}
 
It's not clear how to improve this argument. We should only apply Grover's algorithm to remove a quantifier, when the quantifier being removed contains $O(\log n)$ bits. Thus, we should only consider applying Grover's algorithm after the application of a speedup. Furthermore, it looks like we must always apply a slowdown immediately after an application of Grover's algorithm, because does not seem to be anything else to do with the quantum verifier. These observations, chained together, gave us \cref{lem:grodown}. However, this does not work when the verifier is $\BPTS$ rather than $\TS$. Let's try doing the same thing. 

\begin{align*}
    (Q_k n^{a_k}) \BPTS[n^d] & \subseteq (Q_k n^{a_k}) (Q_{k+1} n^x) (Q_{k+2} x\log n) \TS[n^{d-x}] \\
    & \subseteq (Q_k n^{a_k}) (Q_{k+1} n^x) \BQTIME[n^{d-\f{x}{2}}] \\
    & \subseteq (Q_k n^{a_k}) \BPTS[n^{c\cdot\max(x,d-\f{x}{2})}]. \\
\end{align*}

Even after taking $x = \f{2}{3}$ to minimize the final expression, we are worse off than we started if $c > 1.5$. 

\bibliography{qma-lb-full}
\newpage 

\appendix

\section{Limits of Alternation-Trading Proofs}
\label{sec:limits}

In this section, we will prove the optimality of \cref{thm:gen-lb} among alternation-trading proofs using \cref{cor:speedup-rule} and a generic slowdown rule.

\begin{theorem}
    \label{thm:williams-optimal}
    Fix $0 < \alpha \leq 1$. There is no alternation-trading proof using the standard speedup rule and a generic slowdown rule with parameter $\alpha$ that derives a contradiction for any $c > r_1$, where $r_1$ is the largest root of the polynomial $P_{\alpha}(x) := \alpha^2x^3 - \alpha x^2 -2\alpha x +1$.
\end{theorem}

We prove \cref{thm:williams-optimal} by showing that no annotation and set of speedup parameters can do better than the alternation-trading proof used in \cref{thm:gen-lb}, which we now recall. 

\begin{definition}
    Let $\alpha > 0$ and $c \in (0, \f{\alpha + 1}{\alpha})$. Given the usual speedup rule and a generic slowdown rule with parameter $\alpha$, the \textbf{Good proof} of height $k$ is the alternation-trading proof $\boldsymbol{1_{x_1}1_{x_2}\dots 1_{x_k}0(20)^k}$ where $x_i := \paren{c(\alpha c-1)}^{i-1}\f{a}{\alpha c^2}$ for $i \geq 1$. 
\end{definition}

\begin{corollary}[of \cref{thm:gen-lb}]
    \label{cor:opt-2proper}
    Fix $0 < \alpha \leq 1$. Let $r_1$ be the largest root of the polynomial $P_{\alpha}(x) := \alpha^2x^3 - \alpha x^2 -2\alpha x +1$. Alternation-trading proofs with annotations of form $\boldsymbol{1_{x_1}1_{x_2}\dots1_{x_k}0(20)^k}$ using a generic slowdown rule with parameter $\alpha$ and that only apply Rule $2$ properly cannot derive a contradiction when $ c \geq r_1$. 
\end{corollary}

\begin{proof}
    This follows from the proof of the \cref{thm:gen-lb}, since the bound we obtain saturates the constraints imposed by this annotation and the condition that every application of Rule $2$ is proper. 
\end{proof}

Observe that the Good proof is thus the unique extremal proof among proofs of the form in \cref{cor:opt-2proper}.

We'll start by showing a simple bound, in the spirit of Theorem 3.3 in \cite{williams-aut}. 

\begin{lemma}
    \label{lem:simple-bound}
    Fix $0 < \alpha \leq 1$. There is no alternation-trading proof using the standard speedup rule and a generic slowdown rule with parameter $\alpha$ that derives a contradiction for $c > \f{1+\alpha}{\alpha}$.
\end{lemma}

By \cref{lem:gen-slowdown-params}, we see that this indeed a weaker result than \cref{thm:williams-optimal}. We will use the following claim in order to prove \cref{lem:simple-bound}.

\begin{claim}
        \label{clm:100}
        No proof with annotation $\boldsymbol{100}$ can derive a contradiction for $c > \f{\sqrt{1+\alpha}}{\alpha}$.
\end{claim}

\begin{proof}[Proof of \cref{clm:100}]
    An annotation of $\boldsymbol{1_x00}$ corresponds to the proof
    \begin{align*}
        \TS[n^d] & \subseteq (\exists n^x) (\forall x\log n) \TS[n^{d-x}] \\
        & \subseteq (\exists n^x) \TS[n^{\max(\alpha c(d-x), cx}] \\
        & \subseteq \TS[n^{\alpha^2 c^2 (d-x), \alpha c^2x}].
    \end{align*}

    We want to make the final verifier runtime as small as possible, so we should take $x = \f{\alpha}{\alpha+1}d$ to equalize the terms in the maximum. Then, setting $d = \alpha c^2x = c^2 \f{\alpha^2}{\alpha + 1}d$, we obtain a bound against $c < \f{\sqrt{\alpha + 1}}{\alpha}$. By our choice of $x$, we know that no value of $x$ can do better than this one.
\end{proof}

\begin{proof}
    This proof is due to \cite{williams-aut}. Suppose for contradiction that the lemma is false. Observe that any proof annotation of length exceeding three must contain the sequence $\boldsymbol{1_x0}$ somewhere. Suppose the verifier runtime just before this sequence is $d$. Then, $\boldsymbol{1_x0}$ effects the transformation
    \begin{align*}
        (Q_1 n^{a_1}) \dots (Q_k n^{a_k}) \TS[n^d] & \subseteq (Q_1 n^{a_1}) \dots (Q_k n^{a_k}) (Q_{k+1} n^{x}) \TS[n^{d-x}] \\
        & \subseteq (Q_1 n^{a_1}) \dots (Q_k n^{a_k}) \TS[n^{\max(\alpha c(d-x),cx)}].
    \end{align*}

    We claim that we can remove $\boldsymbol{1_x0}$ because $d \leq \max(\alpha c(d-x),cx)$. Suppose not, in which case $d > cx$ and $d > \alpha c(d-x) \iff \f{d}{\alpha} > c(d-x)$. Then, adding these constraints together,
    \begin{equation*}
        \f{\alpha + 1}{\alpha}d > cd \iff \f{\alpha+1}{\alpha} > c,
    \end{equation*}
    yielding the desired contradiction.

    Under this assumption, we can therefore reduce the length of any annotation to three, which means it is of form $\boldsymbol{100}$. The result then follows from \cref{clm:100}
\end{proof}

In order to prove a better bound than this, we will show that every alternation-trading proof is subsumed by a proof of the same form as the Good proof. Then, the result follows from \cref{cor:opt-2proper}. Before we specify our approach, it will help to recall the notion of an annotation graph from \cref{def:annotation-graph}. To simplify our diagrams, we will omit the oscillations due to applications of the squiggle rule, and only mark changes in height corresponding to consecutive applications of speedup or slowdown rules.  

\begin{definition}
    A \textbf{flattened annotation graph} for is an annotation graph with points $(t,h(t))$ such that 
    \begin{itemize}
        \item $h(0) = 0$
        \item $h(i+1) = h(i) + 1$ if the $i^{\t{th}}$ and $(i+1)^{\t{th}}$ steps in the annotation are $1$
        \item $h(i+1) = h(i) - 1$ if the $i^{\t{th}}$ and $(i+1)^{\t{th}}$ steps in the annotation are $0$
        \item $h(i+1) = h(i)1$ otherwise.
    \end{itemize}
\end{definition}
    
The main structures of interest will be maximal substrings of proofs that are Dyck paths (\cref{def:dyck-path}) on the flattened annotation graph. We distinguish between two types: dromedaries, which go ``straight up and down'' (possibly with squiggle rules applied on the way down) and bactrians, which are Dyck paths returning to the starting height several times. Ultimately, we will show that all dromedaries are subsumed by the Good proof and that all bactrians are subsumed by dromedaries. Let us formally define these. 

\begin{definition}
    A \textbf{camel} is a substring of a proof annotation that is a Dyck path on the flattened annotation graph. The \textbf{base} of a camel is the height immediately before the camel is applied.

    A \textbf{dromedary} is a camel of form $\boldsymbol{1^*0(20)^*}$. We say that a dromedary is \textbf{useful} if the verifier runtime after the dromedary is strictly less than the verifier runtime before the dromedary. A \textbf{dromedary proof} is an alternation-trading proof that consists of a single dromedary.

    A \textbf{bactrian} is a camel that is not a dromedary. We say that a bactrian can be \textbf{dromedarized} if it can be replaced with a dromedary without weakening the proof. 
\end{definition}

Note that any dromedary proof is of the form $\boldsymbol{1_{x_1}1_{x_2}\dots 1_{x_k} 0^+(1_{w_1}0^+)(1_{w_2}0^+) \dots}$.

\begin{figure}[H]
    \centering
    \begin{tikzpicture}
        \begin{axis}[
            xlabel={Step},
            ylabel={Height},
            xmin=-1, xmax=26,
            ymin=-1, ymax=7,
            ticks=none,
            unit vector ratio = 1 1
        ]

        \addplot[
            color=black,
            mark=*,
            ]
            coordinates {
            (0,0)(1,2)(2,3)(3,4)(4,5)(5,4)(6,3)(7,4)(8,5)(9,6)(10,5)(11,4)(12,5)(13,4)(14,5)(15,4)(16,5)(17,4)(18,3)(19,2)(20,3)(21,4)(22,3)(23,2)(24,1)(25,0)
            };
        \draw [thick, blue] (1.5,2.5)--(18.5,2.5);
        \draw [thick, blue] (18.5,2.5)--(18.5,6.5);
        \draw [thick, blue] (18.5,6.5)--(1.5,6.5);
        \draw [thick, blue] (1.5,6.5)--(1.5,2.5);
        \end{axis}
    \end{tikzpicture}
    \caption{A bactrian in an annotation graph}
    \label{fig:bactrian}

\end{figure}

\begin{figure}
    \centering
    \begin{minipage}{0.45\textwidth}
        \centering
        \includegraphics[width = \textwidth, height=0.7\textwidth]{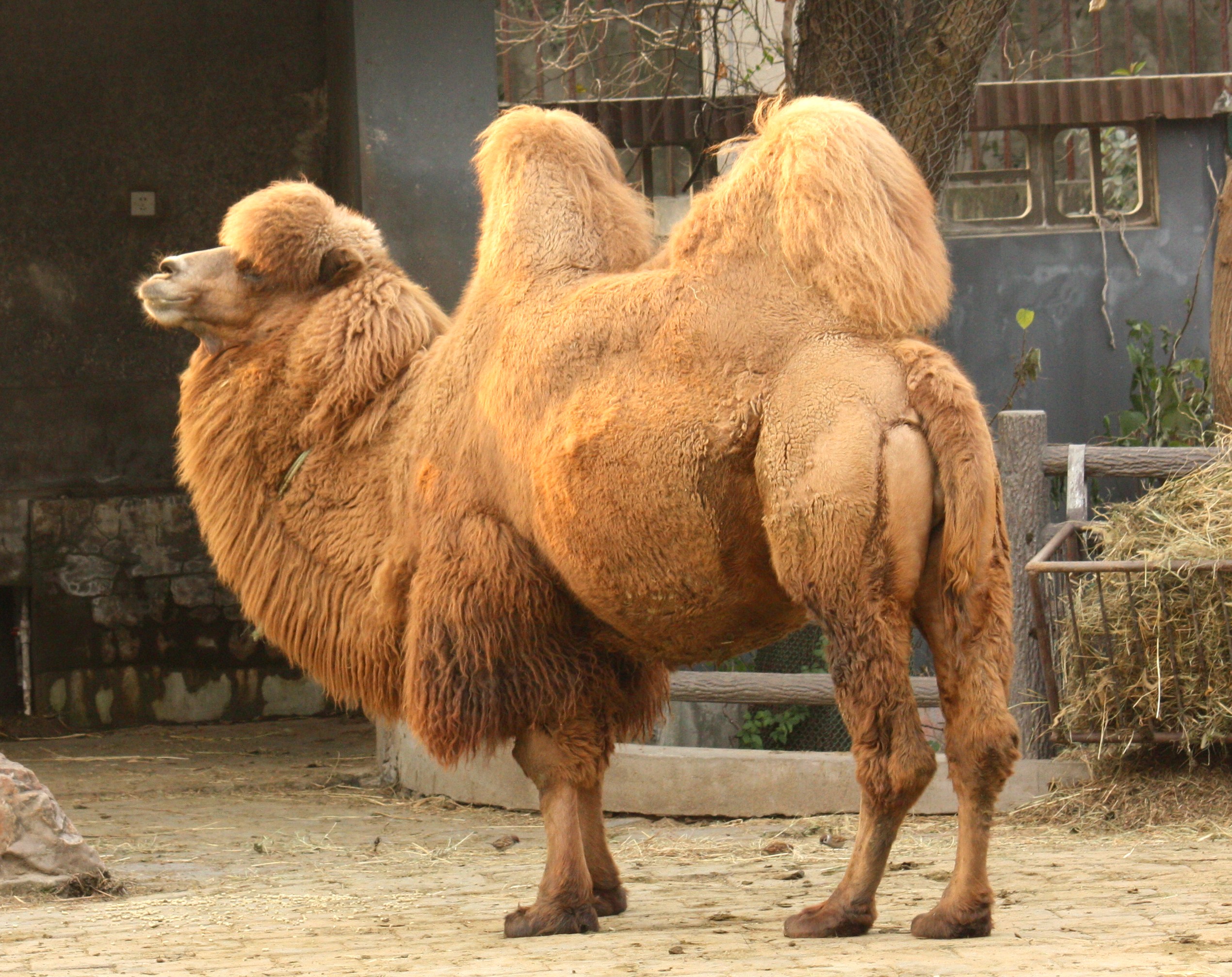}
    \end{minipage}\hfill
    \begin{minipage}{0.45\textwidth}
        \centering
        \includegraphics[width = \textwidth, height=0.7\textwidth]{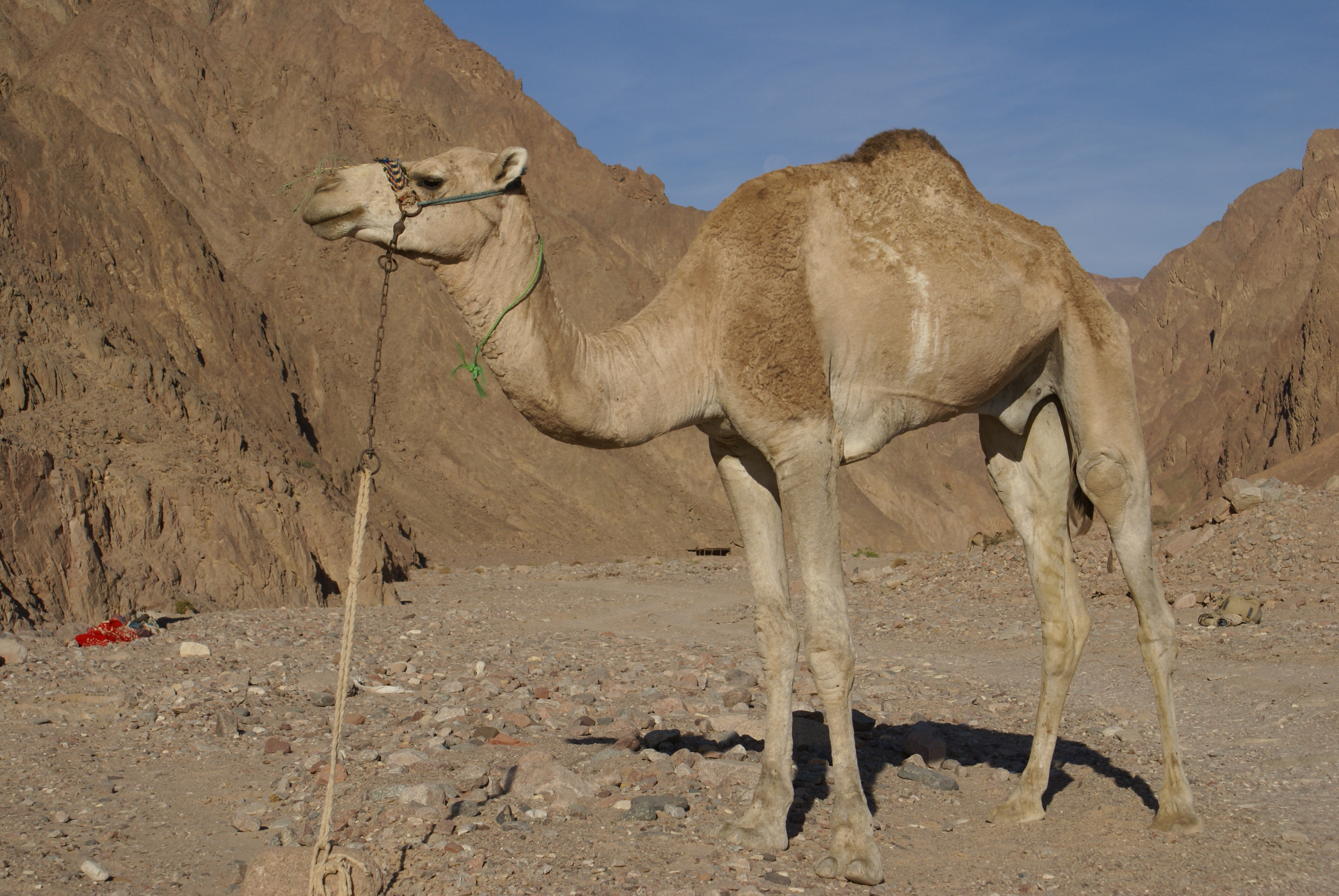}
    \end{minipage}\hfill
    \caption{A bactrian camel (left) and dromedary camel (right). Images taken from Wikipedia and used under a CC BY-SA 3.0 License.}
\end{figure}

\cref{thm:williams-optimal} follows immediately from the following two results. 

\begin{restatable}{corollary}{bactodrom}
    \label{cor:bac2drom}
    Any alternation-trading proof containing bactrians is subsumed by a dromedary proof. 
\end{restatable}

This says that bactrians can always be replaced by dromedaries\footnote{The authors of this paper do not endorse camel monoculture in contexts outside of complexity theory.}.

\begin{restatable}{lemma}{singleblock}
    \label{lem:single-block}
    If there is a dromedary proof yielding a bound for $c_0$, then the Good proof yields a bound for $c \geq c_0$. 
\end{restatable}

\subsection{\texorpdfstring{\cref{lem:single-block}}{Lemma A.9}: The Good Proof is the Best Dromedary Proof}

\begin{definition}
    Let $\cA$ be a valid annotation of form $\boldsymbol{1^k0^+(10^+)^*}$. The \textbf{block decomposition} of $\cA$ breaks the annotation into $k$ disjoint blocks $B_0, B_k, \dots, B_1$ as follows. 
    \begin{itemize}
        \item $B_0 = 1^k0$
        \item $B_{k-i+1}$ is the $i^{\t{th}}$ block of form $\boldsymbol{(10)^*0}$ following $B_0$. 
    \end{itemize}
\end{definition}

    Intuitively, for $i > 0$, the block $B_i$ corresponds to the portion of the annotation where the $i^{\t{th}}$ quantifier (from the left) is removed. 
\begin{example}
    The annotation
    \begin{equation*}
        \boldsymbol{111110001010101010100101010100100}
    \end{equation*}
    breaks into the blocks
    \begin{equation*}
        b_0 = \boldsymbol{111110}, b_5 = \boldsymbol{0}, b_4 = \boldsymbol{0}, b_3 = \boldsymbol{1010101010100}, b_2 = \boldsymbol{101010100}, b_1 = \boldsymbol{100}.
    \end{equation*}
\end{example}

\begin{proof}[Proof of \cref{lem:single-block}]
    Suppose the initial proof is $\cP_0 := \boldsymbol{1_{x_1}1_{x_2}\dots 1_{x_k} 0^+(1_{w_1}0^+)(1_{w_2}0^+) \dots}$, and let $X_i$ be the set of speedup parameters contained in block $B_i$ of the block decomposition of $\cA$. 

    \begin{claim}
        \label{clm:block-max}
        If we set $y_i := \max(x_i, X_i)$, then the proof $\cP_1 := 1_{y_1}1_{y_2}\dots 1_{y_k}0^+(1_{w_1}0^+)(1_{w_2}0^+) \dots$ yields a bound $c_1 \geq c_0$. 
    \end{claim}
    
    \begin{proof}[Proof of \cref{clm:block-max}]
        This claim follows from two observations. 
        \begin{itemize}
            \item When we increase one of the $x_i$ in a dromedary, we don't have to pay for it until the start of block $B_i$, when we have to remove a quantifier with larger exponent than we had to remove initially --- until then, increasing $x_i$ only helps us, by reducing the verifier runtime. 
            \item In the original proof, after block $B_i$, the verifier runtime is at least $c\max(x_i, X_i)$. To see this, note that before we even apply anything from $B_i$, our verifier runtime exponent is at least $cx_i$, and that applying $\boldsymbol{1_w0}$ to $(Q_1 n^{x_1}) \dots (Q_i n^{x_i}) \TS[n^a]$ yields a verifier runtime exponent that is at least $cw$, as we squash the quantifiers of the same type together. Written out,
                \begin{align*}
                    (Q_1 n^{x_1}) \dots (Q_i n^{x_i}) \TS[n^d] & \subseteq (Q_1 n^{x_1}) \dots (Q_i n^{\max(x_i,w)}) (Q_{i+1} w\log n) \TS[n^{d-w}] \\
                    & \subseteq (Q_1 n^{x_1}) \dots (Q_i n^{\max(x_i,w)}) \TS[n^{\max(\alpha c(d-w), cx_i, cw)}]. \\
                \end{align*}
        \end{itemize}

        The point here is that if $w > x_i$ then after applying $1_w0$ the exponent in quantifier $Q_i$ will increase to $w$. If we are going to end up with a larger exponent anyways, we may as well take advantage of the benefits of having a larger speedup exponent by using $w$ from the beginning. 
    \end{proof}

    Next, we will show that we can take all the values in $X_i$ to equal $y_i$, since this is what we use in squiggle rule (\cref{def:rule-2}).

    \begin{claim}
        \label{clm:rule-2-improper}
        Given an alternating proof $\boldsymbol{1_{y_1}1_{y_2}\dots 1_{y_k}0^+(1_{w_1}0^+)(1_{w_2}0^+) \dots}$ satisfying $y_i \geq \max(X_i)$ and yielding a bound $c_1$, the alternating proof $\boldsymbol{1_{y_1}1_{y_2}\dots 1_{y_k}0(20)^k}$ yields a bound $c_2 \geq c_1$. 
    \end{claim}

    \begin{proof}
        Observe that applying $1_w0$ yields 
        \begin{align*}
            (Q_1 n^{y_1}) \dots (Q_i n^{y_i}) \TS[n^d] & \subseteq (Q_1 n^{y_1}) \dots (Q_i n^{\max(y_i,w)}) (Q_{i+1} w\log n) \TS[n^{d-w}] \\
            & \subseteq (Q_1 n^{y_1}) \dots (Q_i n^{\max(y_i,w)}) \TS[n^{\max(\alpha c(d-w), cy_i, cw)}]. \\
        \end{align*}

        We know that $w \leq y_i$ for any $w \in B_i$ because we've already applied \cref{clm:block-max} and set $y_i = \max X_i$. We thus can only do better by making $w = y_i$.

        This tells us that our definition of the squiggle rule in \cref{def:rule-2} in fact used the best possible speedup parameters! Each block $B_i$ for $i > 0$ can thus be replaced by $(20)$. 
    \end{proof}
    
    Next, we need to show that each application of the squiggle rule can be made proper. 

    \begin{claim}
        \label{clm:rule-2-proper}
        Given an alternating proof $\boldsymbol{1_{y_1}1_{y_2}\dots 1_{y_k}0(20)^k}$ yielding a bound $c_2$, there exists a choice of parameters $z_i$ such that the alternating proof $\boldsymbol{1_{z_1}1_{z_2}\dots 1_{z_k}0(20)^k}$ only applies Rule $2$ properly and yields a bound $c_3 \geq c_2$. 
    \end{claim}

    \begin{proof}
        Suppose that the first time we are unable to apply Rule $2$ properly is at block $B_v$. Because we were able to apply Rule $2$ properly until then, we know that the class right before $B_v$ is
        \begin{equation*}
            (Q_1 n^{y_1}) \dots (Q_v n^{y_v}) \TS[n^{\alpha c^2y_{v+1}}].
        \end{equation*}

        Furthermore, because we cannot properly apply Rule $2$, we know that 
        \begin{equation*}
            \alpha c^2y_{v+1} > \f{\alpha c}{\alpha c -1 }y_v \iff y_{v+1} > \f{1}{c(\alpha c -1)}y_v.
        \end{equation*}

        As we continue, there are two possible cases. The first is that at some index $u < v$, we are able to again apply Rule $2$ properly, and the second is that we continue to apply Rule $2$ improperly all the way until the end of the proof.

        \textbf{Case 1 (Improper for all blocks $B_w$ such that $0 < u < w \leq v$):} We know that we are forced to app:y Rule $2$ improperly for every index $w$ such that $u < w \leq v$. Therefore, 

        \begin{equation*}
            (\alpha c)^i y_{v+1} > \f{1}{c(\alpha c - 1)}y_{v-i} \iff y_{v+1} > \paren{c(\alpha c -1)}^{-1}(\alpha c)^{-i}y_{v-i}
        \end{equation*}
        for $0 \leq i < v-u$, where in the left-hand-side we used that each of the $i$ improper applications of slowdown multiplies the verifier runtime exponent by $\alpha c$. At index $u$, we know that we are able to apply Rule $2$ properly, so
        \begin{equation*}
            \label{eq:improper}
            (\alpha c)^{v-u} y_{v+1} > \f{1}{c(\alpha c - 1)}y_{u} \iff y_{v+1} < \paren{c(\alpha c -1)}^{-1}(\alpha c)^{-(v-u)}y_{u}.
        \end{equation*}

        Dividing, we find that 
        \begin{equation*}
            y_{v-i} < (\alpha c)^{-(v-u-i)}y_u.
        \end{equation*}

        Recall that the scaling factor in \cref{thm:gen-lb} was $\f{1}{c(\alpha c - 1)}$, and by assumption we have $c < \f{1+\alpha}{\alpha} \iff \alpha c > c(\alpha c - 1)$. Therefore, 
        \begin{equation*}
            y_{v-i} < (\alpha c)^{-(v-u-i)}y_u < (c(\alpha c - 1))^{-(v-u-i)}y_u.
        \end{equation*}

        Thus, each of the $y_w$ is \textit{smaller} than it can be --- seeing as we are improperly applying the squiggle rule anyways, we lose nothing by increasing it. Define new values $z_{u+i} := (c(\alpha c -1))^{-i}y_u$ for $0 \leq i \leq v-u$, and replace the speedup parameters $y_{u+i}$ with the new parameters $z_{u+i}$. Note that up until we are about to apply $B_v$, the increased speedup parameters only reduced the runtime verifier exponents that we see. Therefore, $B_v$ is still the earliest place that we may be forced to apply Rule $2$ improperly.

        If $y_{v+1} < \f{1}{c(\alpha c - 1)}z_v$ then we have resolved the issue, and the applications of Rule $2$ are all proper until after $B_u$. Otherwise, we have
        \begin{equation*}
            \label{eq:oneway}
            y_{v+1} > \f{1}{c(\alpha c - 1)}z_v \implies y_{v+1} > (c(\alpha c -1))^{-(v-u+1)}z_u.
        \end{equation*}

        On the other hand, from \eqref{eq:improper}, we know that

        \begin{equation*}
            \label{eq:theotherway}
            y_{v+1} < \paren{c(\alpha c -1)}^{-1}(\alpha c)^{-(v-u)}y_{u} \implies y_{v+1} < (c(\alpha c -1))^{-(v-u+1)}y_u < (c(\alpha c -1))^{-(v-u+1)}z_u,
        \end{equation*}
            
        where in the second inequality we again used $\alpha c > c(\alpha c -1)$. Putting together \eqref{eq:oneway} and \eqref{eq:theotherway} yields a contradiction. Thus, our new selection of speedup parameters makes all the applications of the squiggle rule until $B_u$ proper.

        \textbf{Case 2 (Improper for all blocks $B_w$ such that $0 < w \leq v$):} In this case, we set $z_{1+w} = (c(\alpha c -1))^{-w}\f{d}{\alpha c^2}$ for $0 \leq w \leq v-1$ and $d$ equal to runtime of the first class in the proof. We again want to show that we never decrease speedup parameters and that we have a proper application of the squiggle rule at height $v$. Note that $z_i \geq y_i$ for $i \leq v$. We have the constraint
        \begin{equation*}
            \alpha c^2 (\alpha c)^v y_{v+1} < d
        \end{equation*}
        because after $v$ improper applications of the speedup rule we must reach a contradiction. However, this means that
        \begin{equation*}
            y_{v+1} < \f{d}{\alpha c^2}(\alpha c)^{-v} <  \f{d}{\alpha c^2}(c(\alpha c-1))^{-v} = (c(\alpha c -1))^{-1}z_v,
        \end{equation*}
        so $B_v$ now has a proper application of the speedup rule. By our choice of $z_i$, we are proper all the rest of the way down.

        In either case, the result is an alternating proof using only proper applications of Rule $2$. We've increased speedup parameters, but always in a way that never increases the verifier runtime exponent at any level. 
    \end{proof}

    Combining \cref{clm:block-max}, \cref{clm:rule-2-improper}, and \cref{clm:rule-2-proper}, we start with a generic proof $\cP_0$ of the form in the statement of the lemma and convert it to a proof of form $\boldsymbol{1^k0(20)^k}$ with only proper invocations of Rule $2$ without weakening the bounds we can prove. By \cref{cor:opt-2proper}, no proof of this form can obtain a contradiction for $c \geq r_1$. 
\end{proof}

\subsection{\texorpdfstring{\cref{cor:bac2drom}}{Corollay A.8}: Bactrians can be Replaced with Dromedaries}

Our next step is to show that bactrians can be replaced with dromedaries. Specifically, we wish to prove \cref{cor:bac2drom}, which we now recall. 

\bactodrom*

Our main lemma will be the following. 

\begin{lemma}
    \label{lem:all-or-nothing}
    Suppose there exists a dromedary $\boldsymbol{1_{x_1}\dots 1_{x_k}0(20)^{k-1}}$ applied immediately after a speedup showing that 
    \begin{equation*}
        \label{eq:aon-1}
        \dots (Q_t n^a) (Q_{t+1} a\log n) \TS[n^{d}] \subseteq \dots (Q_t n^a) (Q_{t+1} n^{x_1}) \TS[n^{d'}]
    \end{equation*} 
    for $d' < d$. Then, there exists a dromedary $\boldsymbol{1_{z'_1}\dots 1_{z'_{k'}}0(20)^{k'-1}}$ with $z'_1 = x_1$ that, when followed by an application of Rule $2$, can be applied at the same place to show that 
    \begin{equation*}
        \dots (Q_t n^a) (Q_{t+1} a\log n) \TS[n^{d}] \subseteq \dots (Q_t n^x) (Q_{t+1} n^{x_1}) \TS[n^{cx_1}].
    \end{equation*}
\end{lemma}

The upshot is that if we can use a dromedary $\boldsymbol{1_{x_1}1_{x_2}\dots 1_{x_k}0(20)^{k-1}}$ to reduce the verifier runtime a little bit then we are in fact able to reduce it ``all the way'', to the lowest it can be for this choice of $x_1$.

\begin{remark}
    This may seem surprising, but suppose such a result were not true, and there were dromedaries starting with $\boldsymbol{1_{x_1}}$ that reduced the verifier runtime a little bit even when reducing it all the way to $cx_1$ was not possible. Intuitively, we would be able to exploit this by repeatedly applying this dromedary with all the parameters appropriately scaled down each time until we were in range of the squiggle rule (\cref{def:rule-2}), and achieve something with a bactrian proof that we would not be able to achieve otherwise. Thus, \cref{lem:all-or-nothing} is a necessary prerequisite for \cref{thm:williams-optimal}. This lemma is thus motivated given that we have the benefit of hindsight \cite{williams-buss} and experimental results suggesting that bactrians are indeed not helpful. 
\end{remark}

\begin{proof}[Proof of \cref{cor:bac2drom} assuming \cref{lem:all-or-nothing}]
    First, observe that we can replace any bactrian applied immediately after a speedup with a dromedary. By definition, the bactrian starts with a dromedary. \cref{lem:all-or-nothing} shows that either this dromedary is not useful, in which case it should be removed, or that it is so useful that the rest of the bactrian can be removed. Either way, we remove at least one dromedary of the bactrian. We can iterate until the bactrian is replaced with a dromedary. We can iteratively apply this dromedarizing procedure to highest bactrian until no bactrians remain, at which point we are left with a dromedary proof. 
\end{proof}

\begin{proof}[Proof of \cref{lem:all-or-nothing}]
    Ultimately, we want to show that we can make the verifier runtime after the dromedary less than $\f{\alpha c}{\alpha c -1}x_1$, at which point a final application of the squiggle rule yields the lemma. We will show that we can replace the existing speedup parameters $\{x_i\}$ with new speedup parameters $\{z_i\}$ such that this holds. We will perform the transformation in three steps. First, \cref{clm:drom-proper-squiggles} will show that in any dromedary with parameters $\{x_i\}_{1 \leq i \leq k}$, there exist parameters $\{y_i\}_{1 \leq i \leq k}$ that make every application of the squiggle rule proper with $y_1 = x_1$.  Looking back to \eqref{eq:aon-1}, this means that we may not (yet) have $d' < \f{\alpha c }{\alpha c -1}y_1$ (i.e., we may not yet be in range for the final squiggle rule), but we will have $d' = \alpha c^2 y_2$. Then, \cref{clm:drom-williams-annot} shows that we can choose $z_i$ for $i \geq 2$ as if we were performing a Good proof, again leaving $z_1 = y_1 = x_1$. Intuitively, this makes sense because aside from the bottom step, we are performing the steps $\boldsymbol{1^k0(20)^k}$ with proper squiggle rule applications everywhere, just like in the Good proof. The advantage of this is that we know that we can get arbitrarily close to satisfying all the constraints in the Good proof, as per the proof of \cref{thm:gen-lb}. We may assume that
    \begin{equation}
        \label{eq:theass}
        \f{\alpha c}{\alpha c -1}z_1 < \alpha c^2 z_2 < d.
    \end{equation}

    If the first inequality in \eqref{eq:theass} does \textit{not} hold then after this dromedary, we are ``as good as we can get'', insofar as we can apply the squiggle rule to bring the verifier runtime down to $cz_1$. There's no point in doing a Dyck path after this, as we can never decrease the verifier runtime below this, and the exponent in the ultimate quantifier can only increase. Thus, if the first inequality does not hold then we can replace the bactrian with its first dromedary and be done.

    If the second inequality in \eqref{eq:theass} does not hold then it means that the first dromedary was not helpful. This is because the transformation we've done via \cref{clm:drom-proper-squiggles} and \cref{clm:drom-williams-annot} never increased the final verifier runtime. As such, if it's not helpful then it never was. Therefore, the first dromedary should be removed and we can move to the next dromedary in this bactrian.

    as the first constraint corresponds to not being able to apply the squiggle rule at the end (if we could we would already be done) and the second constraint corresponds to this dromedary being useful. We will be able to use ``pretty-much-tightness'' of \eqref{eq:setone} for the original parameters $\{z_i\}$ to show the existence of another set of parameters $\{z'_i\}$ such that $\alpha c^2 z'_2 < \f{\alpha c}{\alpha c -1}x_1$. This last part of the argument is the most involved and is covered in \cref{clm:drom-all-or-nothing}.

    \begin{claim}
        \label{clm:drom-proper-squiggles}
        In any dromedary with parameters $\{x_i\}_{1 \leq i \leq k}$, there exist parameters $\{y_i\}_{1 \leq i \leq k}$ that make every application of the squiggle rule proper, with $y_1 = x_1$.
    \end{claim}

    \begin{proof}
        Recall \cref{clm:rule-2-proper}, which we used earlier to prove that all the squiggle rule applications in a dromedary proof ought to be proper. We will use the same notation and casework. Case $1$ goes through in exactly the same way, so all we have to worry about is Case $2$. Recall that in this case, the first time we are unable to apply the squiggle rule properly is at block $B_v$, at which point every subsequent application of the squiggle rule is improper. In the proof of \cref{clm:rule-2-proper}, we used the fact that the dromedary in question was the entire proof --- because of this, we were able to set parameters that worked based on $d' < d$, the verifier runtime we knew we had to reach at the end. Here, we do not have this, and so we have no control over what happens ``at the bottom''. As such, we will instead set parameters based on the ``top'', ensuring that we have only proper applications of the squiggle rule until the bottom. More precisely, we will set $y_w := (c(\alpha c -1))^{v-w+1}x_{v+1}$ for $2 \leq w \leq v$. The proof goes through in the same way as \cref{clm:rule-2-proper}. 
    \end{proof}

    \begin{claim}
        \label{clm:drom-williams-annot}
        After we've applied \cref{clm:drom-proper-squiggles} and gotten speedup parameters $\{y_i\}$, we may use the parameters $z_{2+i} := \paren{\f{1}{c(\alpha c -1)}}^iy_2-\eps$ for $\eps = O\paren{\f{1}{k^2}}$ for $i \geq 0$ to obtain a bound that's no worse.  
    \end{claim}

    Observe that these are the parameter scalings from \cref{thm:gen-lb}.
    \begin{proof}
        In light of \cref{clm:drom-proper-squiggles}, our dromedary contains a dromedary effecting, for some $a$, the inclusion
        \begin{equation*}
            \label{eq:idk}
            \dots (Q n^{y_2}) (\neg Q y_2\log n) \TS[n^d] \subseteq (Q n^{y_2}) \TS[n^{cy_2}] 
        \end{equation*}

        This is because we know that the squiggle rule applications are proper. However, observe that this is the sort of inclusion that we used in proving \cref{thm:gen-lb}. By \cref{cor:opt-2proper}, we already know the extremal proofs of form $\boldsymbol{1^k0(20)^k}$. If our claim were false and \eqref{eq:idk} was provable even when a sequence of the form in \cref{thm:gen-lb} could not do it, we would be able to contradict \cref{cor:opt-2proper}.
    \end{proof}

    Coupling \cref{clm:drom-proper-squiggles} and \cref{clm:drom-williams-annot}, we have reduced to the case of dromedaries of form $\boldsymbol{1_{z_1}1_{z_2}\dots 1_{z_k}0(20)^{k-1}}$ for $z_i = \paren{\f{1}{c(\alpha c -1)}}^{i-2}z_2$ for $i \geq 2$, proving the inclusion
    \begin{equation*}
        \dots (Q_t n^a)(Q_{t+1} a\log n) \TS[n^d] \subseteq (Q_t n^a)(Q_{t+1} n^{z_1}) \TS[n^{\alpha c^2 z_2}] 
    \end{equation*}
    for $z_1 = x_1$.

    Furthermore, after $1_{z_1}$, the remainder of the dromedary looks like a substring of the optimal proof from \cref{thm:gen-lb}. 

    \begin{claim}
        \label{clm:drom-all-or-nothing}
        Suppose that we can replace the parameters of a dromedary with parameters $\{z_i\}_{i \in [k]}$ of the form in \cref{clm:drom-williams-annot}. Then there exists a $k'$ and parameters $\{z'_i\}_{i \in [k']}$ such that $z'_1 = z_1$ and $\boldsymbol{1_{z'_1}\dots 1_{z'_{k'}}0(20)^{k'-1}}$ proves that 
        \begin{equation*}
            \dots (Q_t n^a)(Q_{t+1} a\log n) \TS[n^d] \subseteq (Q n^a)(\neg Q n^{y_1}) \TS[n^{\alpha c^2 z'_2}]
        \end{equation*}
        and $\alpha c^2 z'_2 < \f{\alpha c}{\alpha c - 1}z_1$.
    \end{claim}

    Thus, after \cref{clm:drom-all-or-nothing} a single application of Rule $2$ is enough to yield a verifier runtime of $cz'_1 = cx_1$ and prove \cref{lem:all-or-nothing}. 

    \begin{proof}
        We will take $z'_2 := (c(\alpha c-1)(1+\epsilon))^{-1}z_1$ and $z'_{2+i} := \paren{c(\alpha c -1)+\epsilon}^{-i}z'_2$ for $\epsilon$ to be set later. This certainly satisfies the second inequality in \eqref{eq:theass}, so it remains to show that for large enough $k'$ the dromedary $\boldsymbol{1_{z_1}1_{z'_2}\dots 1_{z'_{k'}}0(20)^{k'-1}}$ does indeed prove the desired containment. By our choice of speedup parameters, the only constraint that the $\{z'_i\}$ don't automatically satisfy is 
        \begin{equation*}
            (\alpha c)(d-\sum_{i=1}^{k'} z'_i) < \f{\alpha c}{\alpha c -1}z'_{k'},
        \end{equation*}
        corresponding to \eqref{eq:setone}.

        Applying the same sorts of manipulations as in the proof of \cref{thm:gen-lb}, we obtain the equivalent (for these speedup parameters) constraints
        \begin{align*}
            d-z'_1 & < z_2'\paren{\f{1}{\alpha c -1}\paren{c(\alpha c -1)+\epsilon}^{-(k-2)} + \f{1-\paren{c(\alpha c -1)+\epsilon}^{-(k-1)}}{1-\paren{c(\alpha c -1)+\epsilon}^{-1}}} \\
            & = z_2'\paren{c\paren{c(\alpha c -1)+\epsilon}^{-(k-1)}+\f{\paren{c(\alpha c -1)+\epsilon}^{k-1}-1}{\paren{c(\alpha c -1)+\epsilon}^{k-1}} \cdot \f{\paren{c(\alpha c -1)+\epsilon}}{\paren{c(\alpha c -1)+\epsilon}-1}} \\
            & = z_2'\paren{\paren{c(\alpha c -1)+\epsilon}^{-(k-1)}\paren{c-\f{c(\alpha c -1)+\epsilon}{c(\alpha c-1)+\epsilon-1}} + \f{c(\alpha c -1)+\epsilon}{c(\alpha c-1)+\epsilon-1}}.
        \end{align*}

        Because 
        \begin{equation*}
           c < \f{1+\alpha}{\alpha} \implies c < \f{c(\alpha c -1)}{c(\alpha c -1)-1},
        \end{equation*}
        for fixed $\epsilon$ the constraint becomes easier and easier to satisfy as $k'$ increases. Decreasing $\epsilon$ also makes the constraint easier to satisfy for fixed $k'$. Thus, taking $\epsilon = k'^{-1337}$, we find that this is equivalent to the constraint 
        \begin{equation*}
            \label{eq:yay}
            \paren{\f{c(\alpha c-1)-1}{c(\alpha c -1)}}(a-z'_1) < z'_2 - \delta.
        \end{equation*}
        for some positive $\delta$ that can be made arbitrarily small as $k'$ grows large.

        We will show that this constraint is indeed satisfied by upper bounding the left hand side. To do this, we'll need an upper bound on $d$. Fortunately, we have one! We know by \eqref{eq:theass} that

        \begin{equation*}
             z_2 < \f{a}{\alpha c^2},
        \end{equation*}

        and because the original $\{z_i\}$ satisfied the constraints for a Good proof for $i \geq 2$, 
        \begin{equation*}
            \paren{\f{c(\alpha c-1)-1}{c(\alpha c -1)}}(d-z_1) < z_2.
        \end{equation*}

        Therefore, 
        \begin{align*}
            & \paren{\f{c(\alpha c-1)-1}{c(\alpha c -1)}}(d-z_1) < \f{d}{\alpha c^2} \\
            \iff & \paren{1-\f{1}{c(\alpha c-1)} - \f{1}{\alpha c^2}}d < \f{c(\alpha c-1)-1}{c(\alpha c-1)}z_1 \\
            \iff & \f{\alpha c^2(\alpha c -1)-2\alpha c + 1}{\alpha c}d < (c(\alpha c -1)-1)z_1 \\
            \iff & d < \f{\alpha c^2(\alpha c -1)-\alpha c}{\alpha c^2(\alpha c-1)-2\alpha c+1}z_1.
        \end{align*}

        Now that we have the desired upper bound, we may proceed with the proof. Let's start at the left hand side of  \eqref{eq:yay}, noting that $z'_1 = z_1$. 

        \begin{align*}
            \paren{\f{c(\alpha c-1)-1}{c(\alpha c -1)}}(d-z_1) & < \paren{\f{c(\alpha c-1)-1}{c(\alpha c -1)}}\paren{\f{\alpha c^2(\alpha c -1)-\alpha c}{\alpha c^2(\alpha c-1)-2\alpha c+1}-1}z_1 \\
            & = \paren{\f{c(\alpha c-1)-1}{c(\alpha c -1)}}\paren{\f{\alpha c -1}{\alpha c^2(\alpha c-1)-2\alpha c+1}}z_1. \\
        \end{align*}

        By our definition of $z'_2$, we have $z_1 = c(\alpha c -1)(1+\epsilon)z'_2$, so 
        \begin{align*}
            & < \f{\paren{c(\alpha c-1)-1}(\alpha c-1)}{\alpha c^2(\alpha c-1)-2\alpha c+1}z'_2. \\
        \end{align*}

        We want to bound the right hand side by $z'_2-\delta$ for some $\delta$ that we can make as small as we want. Thus, it suffices to show that 
        \begin{equation*}
            \f{\paren{c(\alpha c-1)-1}(\alpha c-1)}{\alpha c^2(\alpha c-1)-2\alpha c+1} < 1. 
        \end{equation*}

        Let's prove this. 
        \begin{align*}
            & c < \f{1+\alpha}{\alpha} \\
            \implies & \alpha c + c(\alpha c-1) < 0 \\
            \implies & \alpha c^2(\alpha c-1)-\alpha c + c(\alpha c -1)+1 < \alpha c^2(\alpha c -1)-2\alpha c + 1 \\
            \implies & \paren{c(\alpha c-1)-1}(\alpha c-1) < \alpha c^2(\alpha c-1)-2\alpha c+1 \\
            \implies & \f{\paren{c(\alpha c-1)-1}(\alpha c-1)}{\alpha c^2(\alpha c-1)-2\alpha c+1} < 1. 
        \end{align*}
    \end{proof}

    As described earlier, chaining together \cref{clm:drom-proper-squiggles}, \cref{clm:drom-williams-annot}, and \cref{clm:drom-all-or-nothing} proves the lemma. 
\end{proof}

\end{document}